%% file: main.tex
\documentclass[runningheads,USenglish,cleveref, autoref]{llncs}
\usepackage{amsmath}
\usepackage{amsfonts}
\usepackage{amssymb}
\usepackage{graphicx}

\usepackage{verbatim}
\usepackage{xcolor}

\usepackage{todonotes}
\usepackage{enumitem}

\newcommand{\makro}[1] { [\![ #1]\!]}
\newcommand{\stat}[2] { #1\!\! :\!#2 }
\newcommand{\rhs}[1] { \text{rhs}(#1)}

\newcommand{\corner}[1] { \ulcorner\! #1 \!\urcorner}

\begin{document}
	
\title{Definability Results for Top-Down Tree Transducers} 
%
%
\author{Sebastian Maneth\inst{1} \and
	Helmut Seidl\inst{2} \and
	Martin Vu\inst{1}}
\authorrunning{F. Author et al.}
%
\institute{Universit\"at Bremen, Germany \and
	TU M\"unchen, Germany
}
\maketitle              
	
	\begin{abstract}
We prove that for a given deterministic top-down transducer with look-ahead
it is decidable whether or not its translation is definable
(1)~by a linear top-down tree transducer or
(2)~by a tree homomorphism.
We present algorithms that construct equivalent such transducers if they exist.
	\end{abstract}

\input{newIntro}

\input{preliminaries}
\input{linear}

\input{homomorphism}

\input{conclusions}

\bibliographystyle{splncs04}
\bibliography{lit}

\appendix
\input{look-ahead}

\input{appendix}
\end{document}

%% file: newIntro.tex
\section{Introduction}\label{sec:intro}

Tree transducers are fundamental devices that were invented in the 1970's in
the context of compilers and mathematical linguistics.
Since then they have been applied in a huge variety of contexts.
%
The perhaps most 
basic type of tree transducer is the
top-down tree 
transducer~\cite{DBLP:journals/jcss/Thatcher70,DBLP:journals/mst/Rounds70} 
(for short \emph{transducer}).
Even though top-down tree transducers are very well studied,
some fundamental problems about them have remained open.
For instance, given a (deterministic) such transducer, is it decidable whether
or not its translation can be realized by a  linear  transducer?
In this paper we show that indeed this problem is decidable, 
and that in the affirmative case such a linear 
transducer can be constructed. 

In general, it is advantageous to know whether a translation belongs to a smaller class,
i.e., can be realized by some restricted model with more properties and/or better
resource utilization. The corresponding decision problems for transducers, though, 
are rarely studied and mostly non-trivial.
One recent break-through  is the decidability of
 one-way string
transducers within (functional) two-way string 
transducers~\cite{DBLP:journals/lmcs/BaschenisGMP18}.
Even more recently, it has been proven that
 look-ahead  removal for linear deterministic
top-down tree transducers is decidable~\cite{DBLP:conf/icalp/ManethS20}. 
In our case, one extra advantage of linear transducers (over non-linear ones) is that
linear transducers effectively preserve the regular tree languages ---
implying that forward type checking (where a type is a regular tree language)
can be decided in polynomial time. For non-linear transducers on the other hand, 
type checking 
is DEXPTIME-complete \cite{DBLP:journals/ipl/PerstS04,DBLP:journals/ijfcs/Maneth15}.

The idea of our proof uses the canonical earliest normal form 
for top-down tree transducers~\cite{DBLP:journals/jcss/EngelfrietMS09}
(to be precise, our proof even works for transducers with look-ahead
for which a canonical earliest normal form is presented 
in~\cite{DBLP:journals/tcs/EngelfrietMS16}).
A given canonical earliest transducer $M$ 
produces its output at least
as early as any equivalent linear transducer.
From this we can deduce that if $M$ is equivalent to some linear transducer,
then it has two special properties:
\begin{enumerate}
\item $M$ is \emph{lowest common ancestor conform} (\emph{lca-conform}) and
\item $M$ is \emph{zero output twinned}.
\end{enumerate}
Lca-conformity means that if the transducer copies (i.e., processes
an input subtree more than once), then the output subtree rooted at
the lowest common ancestor of all these processing states may not depend 
on any other input subtree. 
Zero output twinned means that in a loop of two states
that process the same input path, no output whatsoever may be produced.
These properties are decidable in polynomial time and if they hold,
an equivalent linear transducer can be constructed.

In our  second result we prove that for a transducer (with regular look-ahead)
it is decidable whether or not it is equivalent to a
\emph{tree homomorphism}, and that in the affirmative case 
such a homomorphism can be constructed. 
In order to obtain this result, we prove that whenever
a transducer  $T$ is equivalent to a homomorphism,
then any subtree of a certain height of any partial output of $T$ is either ground, 
or effectively identical to the axiom of $T$.
This property can again be checked in polynomial time, and if it holds 
a corresponding homomorphism can be effectively constructed.

For simplicity and better readability we consider total transducers (without look-ahead)
in our first result though we remark that the result  can be extended to partial transducers with look-ahead.
All proofs 
for partial transducers with look-ahead are 
technical variations of the proofs for the total case
and can be found in the Appendix.
Note that our results also work  for  given bottom-up tree transducers,
because they can be simulated by transducers with look-ahead~\cite{DBLP:journals/mst/Engelfriet77}.

%% file: preliminaries.tex
\section{Preliminaries}
For $k\in\mathbb{N}$, we denote by $[k]$ the set $\{1,\dots,k\}$.
Let $\Sigma=\{e_1^{k_1},\dots, e_n^{k_n}\}$ be a \emph{ranked alphabet}, where
$e_j^{k_j}$ means that the symbol $e_j$ has
\textit{rank} $k_j$.
By $\Sigma_k$ we denote the set of  all symbols of $\Sigma$ which have rank $k$.
The set $T_\Sigma$ of 
\textit{trees over $\Sigma$} 
consists of all strings of the form 
$a(t_1,\dots, t_k)$, where
$a\in \Sigma_k$, $k\geq 0$, and $t_1,\dots,t_k \in T_\Sigma$. 
We denote by $[a_i \leftarrow t_i \mid i\in [n]]$ the \emph{substitution} 
that replaces each leaf labeled $a_i$ by
the tree $t_i$ (where the $a_i$ are distinct symbols of rank zero
and the $t_i$ are trees).
The set $V(t)$ of  nodes consists of $\lambda$ (the root node) 
and strings  $i.u$, where $i$ is a positive integer and $u$ is a node.
E.g. for the tree $t=f(a,f(a,b))$ we have $V(t)=\{\lambda,1,2,2.1,2.2\}$.
For  $v\in V(t)$, $t[v]$ is the label of $v$, 
$t/v$ is the subtree of $t$ rooted at $v$, and
$t[v \leftarrow t']$ is obtained from $t$ by replacing $t/v$ by $t'$.
 The \emph{size} of a tree $t$, denoted by $|t|$, is its number 
 of nodes, i.e., $|t|=|V(t)|$.

We fix the \textit{set $X$ of variables} as $X=\{x_1,x_2,x_3\dots\}$
and let $X_k=\{x_1,\dots, x_k\}$.
Let $A,B$ be sets.
We let $A(B)=\{a(b)\mid a\in A,b\in B\}$
and $T_\Sigma[S]=T_{\Sigma'}$ where $\Sigma'$ is obtained from $\Sigma$
by $\Sigma'_0=\Sigma_0\cup S$.
The set of \emph{patterns} $T_\Sigma\{X_k\}$ is $T_\Sigma$ plus
all trees in $T_\Sigma[X_k]$ that contain each $x\in X_k$ exactly once,
and in-order of the tree, e.g., $f(a,x_1,f(x_2,x_3))$ is a pattern. 
%
We say that  a tree $t\in T_\Sigma [X_k]$, $k\geq 0$, is a \textit{prefix} of $t'$
 if there are suitable trees  $t_1,\dots,t_k$ such that $t[t_1,\dots, t_k]=t'$, where $t[t_1,\dots, t_k]$ denotes the tree $t[x_i \leftarrow t_i \mid x_i\in X_k]$. By definition, any tree $t$ is a prefix of itself.
For trees $t_1,t_2$, $t_1 \sqcup t_2$ denotes
the \emph{maximal} pattern that is a prefix of   $t_1$ and $t_2$.
E.g. $f(a, g(a)) \sqcup  f(b, b)  = f(x_1,x_2)$
and $ g(g(b)) \sqcup g(a)   =  g(x_1)$.
  
\subsection{Transducers} \label{transducer section}
A \textit{deterministic total top-down tree transducer}  $T$
(or \emph{transducer} for short) is a tuple $T=(Q,\Sigma,\Delta,R, A)$ where
\begin{itemize}
\item $Q$ is a finite set of \emph{states},
\item $\Sigma$ and $\Delta$ are the ranked  \emph{input} and \emph{output alphabets}, respectively, disjoint with $Q$,
\item $R$ is the \emph{set of rules},
\end{itemize}
and
$A\in T_\Delta [Q(X_1)]$ is the \emph{axiom}.
For every $q\in Q$ and $a\in\Sigma_k$, $k\geq 0$,
the set $R$ contains exactly one rule of the form
$q(a(x_1,\dots,x_k)) \to t$, where $t\in T_\Delta [Q(X_k)]$ is also
denoted by $\text{rhs}_T(q,a)$.
The \emph{size} $|T|$ of $T$ is $\sum_{q\in Q,\sigma\in\Sigma}|\text{rhs}_T(q,a)| + |A|$.
We say that  $T$ is \textit{linear} if 
in the axiom and in each right-hand side, every variable occurs at most once.
A transducer $h$ with $|Q|=1$ and axiom of the form $q(x_1)$ is called a \textit{homomorphism}. 
As $h$  has only one state, we write 
$h(a)$ instead of $\rhs{q,a}$.
For $q\in Q$, we denote by $\makro{q}$  the  function from
$T_\Sigma [X]$ to $T_\Delta [Q(X)]$ defined  as follows:
\begin{itemize}
\item 
$\makro{q}(s)=\rhs{q,a}[q(x_i) \leftarrow \makro{q} (s_i) \mid q\in Q, i\in [k]]$
for $s=a(s_1,\dots,s_k)$, $a\in \Sigma_k$, and $s_1,\dots,s_k\in T_\Sigma[X]$ 
\item $\makro{q}(x)= q(x)$ for $x\in X$.
\end{itemize}
We define the function $\makro{T}:T_\Sigma[X] \to
T_\Delta [Q(X)]$ by $\makro{T}(s)=A [q(x_1) \leftarrow \makro{q}(s) \mid q\in Q]$, 
where $s\in T_\Sigma [X]$. For simplicity we also write $T(s)$ instead of  $\makro{T}(s)$.
Two transducers $T_1$ and $T_2$ are \emph{equivalent} if
the  functions $\makro{T_1},\makro{T_2}$ restricted to ground input trees are equal.

\begin{example}\label{transducer example}
Let $\Sigma =\{a^1, e^0 \}$ and $\Delta = \{ f^2, e^0\}$. Consider the transducers 
$T=(Q,\Sigma,\Delta,R, A)$ where $Q=\{q\}$, $A=f(q(x_1), q(x_1))$ and $R$ consists of the rules
\begin{center}
	$\begin{array}{lcl c lcl }
	q (a(x_1)) &\to&  f(q(x_1), q(x_1)) & \text{and} &   q (e) &\to&   e,  \\
	\end{array}$
\end{center}
and $T'=(Q',\Sigma,\Delta,R', A')$
where $Q'=\{q'\}$, $A'=q'(x_1)$ and $R'$ consists of 
 \begin{center}
	$\begin{array}{lcl c lcl }
	q' (a(x_1)) &\to&  f(q'(x_1), q'(x_1)) & \text{and} &   q' (e) &\to&  f( e, e).  \\
	\end{array}$
\end{center}
Clearly, it can be verified that both transducers are equivalent, i.e., both transducers transform a monadic
input tree of height $n$ with nodes labeled by $a$ and $e$ to a full binary tree of height $n+1$
with nodes labeled by $f$ and $e$. Additionally, we remark that 
$T'$ is  a homomorphism.
\end{example}

We say that a  transducer $T$ is \emph{earliest} if $\bigsqcup \{\makro{q}(s)\mid s\in T_\Sigma \}=x_1$ holds for all $q\in Q$. 
Informally this means that
for each state $q\in Q$ there are input trees $s_1,s_2 \in T_\Sigma$ such that $\makro{q}(s_1)$ and $\makro{q} (s_2)$ have different root symbols.
We call $T$ \emph{canonical earliest} if $T$ is earliest and for distinct states $q_1$, $q_2$ it holds that  a ground tree $s$ exists such that $\makro{q_1}(s)\neq \makro{q_2}(s)$.\\
Consider the transducers in Example~\ref{transducer example}. Clearly, $T'$ is not canonical earliest as
 for all input trees $s'$ the root symbol of $\makro{q'}(s')$ is $f$.
The transducer $T$ on the other hand is canonical earliest as  the root symbol of $\makro{q}(e)$ is $e$ while
for all $s\neq e$ the root symbol of $\makro{q}(s)$ is $f$.
We remark that canonicity of $T$ is obvious as $T$  has only one state. 

\begin{proposition}
	\cite[Theorem 16]{DBLP:journals/jcss/EngelfrietMS09}
	For every  transducer $T$ an equivalent  canonical earliest transducer $T'$  can be constructed
	in polynomial time.
\end{proposition}

Intuitively, any   transducer $T'$ that
is equivalent to a canonical earliest  transducer $T$ cannot 
generate ``more'' output than $T$ on the same input.
Therefore the following holds.

\begin{lemma}\label{auxiliary total lemma}
	Let $T$ and $T'$ be equivalent  transducers. Let $T$ be  canonical earliest.
	Let $s\in T_\Sigma [X]$.  Then  $V(T'(s)) \subseteq V(T(s))$ and  if $T'(s)[v]=d \in \Delta$ then 
	$T(s)[v]=d$ for all $v\in V(T'(s))$.
\end{lemma}

%% file: linear.tex
\section{From Transducers  to   Linear Transducers}	
In the following let
$T$ be a  transducer with the same tuple as defined in Section~\ref{transducer section}
 that is canonical earliest.
 We show that it is decidable in polynomial time whether or not
 there is a linear transducer $T'$ equivalent to $T$  and 
 if so how to construct such a linear transducer.
 
 Before we introduce properties which any  canonical earliest transducer that is equivalent to some linear transducer    must have, consider the following definitions.
   We call a non-ground tree  $c\in T_\Sigma \{X_1\}$  a \emph{context}. 
 Let $s$ be an arbitrary tree.
 For better readability we simply write $cs$ instead of $c[x_1\leftarrow s]$. 
 Let $v_1$ and $v_2$ be distinct nodes with labels $q_1(x_1)$ and $q_2(x_1)$, respectively, that occur in ${T}(c)$,
 where $q_1,q_2\in Q$.
 Then we say that  $q_1$ and $q_2$ \emph{occur pairwise in $T$}. 
 
  Next we show a transducer for which no equivalent linear transducer exists.
 
  \begin{example}\label{zo-twined introducrion example}
    Let $T_0$ be a canonical earliest  transducer with axiom
 	$f(q_1(x_1),q_2(x_1))$  and the following rules.
 	\begin{center}
 		$\begin{array}{lcl c lcl }
 		q_{1} (a(x_1)) &\to&  g(q_1(x_1)) &\quad &   q_{2} (a(x_1)) &\to&  q_2(x_1) \\
 		q_{1} (e) &\to&  e & \quad&   q_{2} (e) &\to&  e \\
 		  q_{1} (e') &\to&  e' &\quad &  q_{2} (e') &\to&  e' \\
 		\end{array}$
 	\end{center}
 	 Clearly,  
 	$T_0(a^n (x_1))=f(g^n (q_1(x_1)) ,q_2(x_1) )$ for all  $n\in \mathbb{N}$. 
 	Assume that a linear transducer $T'$   equivalent to $T_0$ exists. 
 	Clearly $T'(a^{n}  (x_1))$ cannot be of the form $f(t_1, t_2)$, where $t_1$
 	and $t_2$ are some trees as otherwise  either $t_1$ or $t_2$ must be ground
 	due to the linearity of $T'$. This contradicts Lemma~\ref{auxiliary total lemma}  as
 	$T_0$ is canonical earliest. 
 	Hence, $T'(a^{n}  (x_1))=q'_n(x_1)$ holds for all  $n\in \mathbb{N}$ where $q'_n$ is some state.
 	Thus, for all $k\in \mathbb{N}$ a ``partial'' input 
 	tree $s$ exists such that 
 	 the height difference of $T_0(s)$ and $T'(s)$ is greater than $k$.
 	It is well known that the height difference of the trees generated by equivalent transducer on the same partial input tree is bounded~\cite{DBLP:journals/ijfcs/Maneth15}.   Hence, $T'$ cannot exist. 
 \end{example}

 The reason why  no  linear transducer  that  is equivalent to the  transducer $T_0$ in Example~\ref{zo-twined introducrion example} exists, is because $q_1$ and $q_2$ occur pairwise and 
 $\makro{q_1}(a(x_1))$ and $\makro{q_2} (a(x_1))$ are ``loops of which at least one
 generates output''. This leads us to our first property, which we call \emph{zero output twinned}.
 
Let $c$ be a context and $q_1$ and $q_2$ be pairwise occurring states.
We write
$(q_1,q_2) {\vdash}^{c} (q_1', q_2')$ if
a node with label $q_i'(x_1)$ occurs in
$\makro{q_i}(c)$ for $i=1,2$  and either $q_1$ or $q_2$ \emph{generate output on input} $c$, i.e.,
either $\makro{q_1}(c) \neq q'_1(x_1)$ or  $\makro{q_2}(c) \neq q'_2(x_1)$ holds.

\begin{definition}
	Let $q_1$ and $q_2$ be   pairwise occurring states of $T$.
We say that $T$ is \emph{zero output twinned} if  no context $c\neq x_1$  exists such that
$(q_1,q_2) {\vdash}^{c} (q_1, q_2)$ holds.
\end{definition}

 In other words, $T$ is zero output twinned if 
 there exists no context $c$  such that 
$\makro{q_1}(c)$ and  $\makro{q_2}(c)$ are loops and at least one of those loops generates output.

\begin{lemma}\label{zero output twinned}
	If $T$ is equivalent to a linear transducer, 
	then $T$ is zero output twinned. 
	It can be decided in time polynomial in $|T|$ whether or not $T$ is zero output twinned.
\end{lemma}

Consider the following transducer. Though this transducer
is zero output twinned  no equivalent linear transducer exists.

\begin{example}\label{lca-conform introducrion example}
	Let $\Sigma=\{a^2,e^0\}$ and $\Delta=\{f^3,e_0^0, e_1^0, e_2^0, e_3^0\}$.
Consider the canonical earliest transducer $T_0=(Q,\Sigma,\Delta,R, A)$.
Let $A=q_0(x_1)$ be the axiom of $T_0$ and
\begin{center}
 $\begin{array}{lcl c lcl }
q_{0} (a(x_1,  x_2)) \to f(q_1( x_1),q_2(x_2), q_3(x_1))&\quad &   q_{0} (e) &\to&  e_0 \\
q_{i} (a(x_1,  x_2)) \to f(e_i,e_i, e_i)& \quad &   q_{i} (e) &\to&  e_i,  \\
\end{array}$
\end{center}
where $i=1,2,3$, be the rules of $T_0$.
Clearly, it holds that $T_0(a(x_1,x_2))= f(q_1(x_1),q_2(x_2), q_3(x_1))$. 
Note that as no loops occur in $T_0$, $T_0$ is zero output twinned.
Assume that a linear transducer $T'$  equivalent to $T_0$ exists.
Analogously to Example~\ref{zo-twined introducrion example}, Lemma~\ref{auxiliary total lemma} and the 
 linearity of $T'$ yield that  either $T'(a(x_1,x_2))=q'(x_1)$ or  $T'(a(x_1,x_2))=q'(x_2)$ where $q'$ is some state.
 
W.l.o.g. consider the former case. Let $s_1, s_2$ be distinct trees. Then 
$\makro{q'}(s_1)=T'(a(s_1,s_2)) = T_0(a(s_1,s_2))= f(\makro{q_1}(s_1),\makro{q_2}(s_2), \makro{q_3}(s_1))$,
which means that the tree generated by $q_2$ cannot depend on its input $s_2$.
This contradicts the earliest property of $T_0$.  Hence, $T'$ cannot exist. 
\end{example}

In the following we show that the situation described in Example~\ref{lca-conform introducrion example} occurs if 
the transducer does not have the following property which we call \emph{lowest common ancestor conform}.

First, we introduce some terminology.
Let $v, v'$ be  nodes. Recall that by
definition nodes are strings. Then $v$ is an \emph{ancestor} $v'$ if $v$ is a prefix of $v'$. 
Furthermore, by the \emph{lowest common ancestor} of nodes $v_1,\dots, v_n$ we refer to the longest common prefix of those
nodes.  
Let $s\in T_\Sigma  \{X_k\}$.
  For $1\leq i\leq k$, we denote by $\nu_T (s,x_i)$ the lowest common ancestor of all leaves of $T(s)$ that have
labels of the form $q(x_i)$ where $q\in Q$. If $T$ and $s$ are clear from context we just write
$\nu (x_i)$.

\begin{definition}
The transducer $T$ is called \emph{lowest common ancestor conform} \emph{(lca-conform} for short\emph{)} if  for an arbitrary  input tree $s\in T_\Sigma \{X_k\}$, $k\in\mathbb{N}$, 
the output  tree ${T}(s)$ is lca-conform. 
An output tree  ${T}(s)$  is \emph{lca-conform} 
if for all $x_i$ such that $\nu_T (s,x_i)$ is defined, no leaf with label of the form $q(x_j)$, $j\neq i$, occurs in $T(s)/ \nu_T (s,x_i)$.
\end{definition}

 For instance,  $T_0$ in Example~\ref{lca-conform introducrion example}  is  not lca-conform because 
 $\nu_{T_0} (a(x_1,x_2)  ,x_1)$ is the root of $T_0(a(x_1,x_2))= f(q_1(x_1),q_2(x_2), q_3(x_1))$ and this tree obviously contains $q_2(x_2)$.
On the other hand, the transducer in Example~\ref{zo-twined introducrion example} is lca-conform because its input trees are monadic.

\begin{lemma}\label{lca-conform}
	If $T$ is equivalent to a linear transducer, then $T$ is  lca-conform.
		It can be decided in time polynomial in $|T|$ whether or not $T$ is lca-conform.
\end{lemma}

\begin{proof}
	We prove  that  $T$ is  lca-conform if $T$ is equivalent to a linear transducer~$T'$.
	For the second part of our statement we refer to the Appendix.
	
	Let   $s\in T_\Sigma \{X_k\}$.
	Let $v$ be a node such that ${T'}(s)[v]= q'(x_i)$ where $q'$ is some state  and $1\leq i \leq k$. 
	First, we show that the subtree ${T}(s)/v$  contains no leaves that have labels of the form $q(x_j)$ with $i\neq j$.	
	Note that due to Lemma~\ref{auxiliary total lemma},  ${T} (s) / v$ is  defined.
	
	Assume to the contrary that  ${T}(s)/v$ contains  some node labeled $q(x_j)$, i.e., ${T}(s)[ \hat{v}]=q(x_j)$ for some descendant
	$\hat{v}$ of $v$.
	Consider the trees ${T'}(s[x_i\leftarrow s'])$ and ${T}(s[x_i\leftarrow s'])$ where $s'$ is some ground tree.
	 Clearly   ${T}(s[x_i\leftarrow s']) [\hat{v}]=q(x_j)$ and hence ${T}(s[x_i\leftarrow s'])/v$ is not ground. However ${T'}(s[x_i\leftarrow s'])/v= \makro{q'} (s')$ is ground. This contradicts Lemma~\ref{auxiliary total lemma}.
	Thus we deduce that for all $v\in V(T'(s))$ if  $T'(s) / v =q'(x_i)$ then  
	only  leaves with labels in $\Delta$ or with labels  of the form $q(x_i)$, $q\in Q$, occur in $T(s)/v$  (*).
	
	We now show that if leaves with label of the form $q(x_i)$ occur in $T(s)$ then $v$ is an ancestor node of all those 
	leaves. Assume to the contrary that some leaf $\tilde{v}$ with label of the form $q(x_i)$ occurs in $T(s)$ that is not a descendant of $v$. Due to Lemma~\ref{auxiliary total lemma} it is clear that either $\tilde{v}\notin V(T'(s))$
	or $T'(s) [\tilde{v}]$ is not labeled by a symbol in $\Delta$. In both cases, some ancestor of $\tilde{v}$ must
	have label of the form $q'(x_j)$ in $T'(s)$ otherwise $T$ and $T'$ cannot be equivalent.
	Note that by definition any node is an ancestor of itself.  As  $\tilde{v}$ is not a descendant of $v$,  ${T'}(s)[v]= q'(x_i)$ and $T'$ is linear, we conclude that $x_i\neq x_j$ which contradicts (*). 
	Thus,  
	we conclude that  $v$  is a common ancestor of all  leaves with label of the form $q(x_i)$ in $T(s)$.
	
	As ${T}(s)/v$  does not contain leaves that have label of the form $q(x_i)$  neither does $T(s)/ \nu (x_i)$
	as $v$ is an ancestor of $\nu (x_i)$.
	 Hence, ${T}(s)$ is lca-conform.  \qed
\end{proof}

\noindent
 The next lemma is used to show that the following construction terminates.
\begin{lemma}\label{lca-tree bound}
	Let $T$ be lca-conform and zero output twinned.
	Let $c$ be a context. Then $\text{height}({T}(c)/ \nu (x_1) )\leq(|Q|^2+1)\eta$, where                     
	 $\eta=\text{max} \{\text{height}(\text{rhs} (q,a)) \mid q\in Q, a\in \Sigma \}$.
\end{lemma}

\subsection{Constructing a  Linear Transducer}\label{construction}
Subsequently, we give a construction which yields  a linear transducer $T'$ equivalent to a given canonical earliest transducer $T$ 
if $T$ is zero output twinned and lca-conform.

Assume in the following that  the axiom $A$ of $T$ is  not ground (the case that $A$ is ground is trivial). 
By Lemma~\ref{auxiliary total lemma} the output of $T$ is ``ahead'' of the output of any equivalent transducer on the same input tree. Therefore, the basic idea is the same as in~\cite{DBLP:journals/tcs/EngelfrietMS16} and \cite{DBLP:conf/icalp/ManethS20}: we store the ``aheadness'' of $T$ compared to $T'$ in the  states of $T'$.
The  states of $T'$ are  of the form $\langle t \rangle$, where $t$ is a tree in $T_\Delta [Q]\setminus T_\Delta$
that has a height bounded according to Lemma 4. 
Furthermore, we demand that the root of $t$
 is the lowest common ancestor of all leaves of $t$ labeled by  symbols in $Q$.
For  such a state $\langle t\rangle$ 
we write ${t} [\leftarrow x]$ instead of $t[q\leftarrow q(x) \mid q\in Q]$
for  better readability. 

We define  $T'$ inductively.
 We define its axiom as $A'=c \langle t_0 \rangle (x_1)$, where $c$ is a context and $t_0\in T_\Delta [Q]$  such that  $c  t_0 [\leftarrow x_1]   =A$. 
Note that $ t_0 [\leftarrow x_1] $ is the subtree  of $A$ rooted at the lowest common ancestor of all nodes
with labels of the form $q(x_1)$.
We now define the rules of $T'$.
Let  $\langle t \rangle$ be a state of $T'$ where $t=p[x_i\leftarrow q_i \mid q_i\in Q, i\in [n]]$ and $p\in T_\Delta [X_n]$. Let $a\in \Sigma_k$ with $k\geq 0$.
Then, we define
$\langle t \rangle (a(x_1,\dots x_k)) \to p'[x_j\leftarrow \langle t_j \rangle (x_j)\mid j\in [k]  ]$
where $p'\in T_\Delta [X_k]$  such that
\[ p'[x_j\leftarrow t_j[\leftarrow  x_j]\mid j\in [k]  ]= p[ x_i \leftarrow \rhs{q_i,a}  \mid i\in [n]]\]
and $t_j [\leftarrow  x_j]$ is the subtree of
$ p[ x_i \leftarrow \rhs{q_i,a}  \mid i\in [n]]$ rooted at the lowest common ancestor of all nodes
 with labels of the form $q(x_j)$. 
In the following we show that  states of $T'$ indeed  store the ``aheadness'' of $T$ compared to $T'$.

\begin{lemma}\label{aux lemma}
	Assume that in our construction some state $\langle t \rangle$  has been defined at some point.
	Then a context $c$   exists such that ${T}(c)/\nu (x_1)=t[\leftarrow x_1]$.
\end{lemma}

Lemmas~\ref{lca-tree bound} and~\ref{aux lemma} yield that
the height of any tree $t$ such that $\langle t \rangle$ is a state of $T'$ is bounded.
Thus, our construction  terminates.
In the Appendix we further show that our construction is well defined (even for partial transducers with regular look-ahead) and that
the transducer of our construction  is indeed equivalent to the given transducer.

The construction and Lemmas~\ref{zero output twinned} and~\ref{lca-conform} yield the following theorems.

\begin{theorem}
	Let $T$ be a  canonical earliest transducer. An equivalent linear transducer $T'$  exists if and only if 
	$T$ is zero output twinned  and lca-conform.
\end{theorem}

\begin{theorem}
	Let $T$ be a transducer. 
	It is decidable in polynomial time whether or not an 
	equivalent linear transducer $T'$  exists, 
	and if so $T'$ can be effectively constructed. 
\end{theorem}

Using the results of \cite{DBLP:conf/icalp/ManethS20} we  show in the appendix that our problem
is decidable even for \emph{partial transducer with regular look-ahead}.

\begin{theorem}\label{linear main result}
	Let $M$ be a  transducer with regular look-ahead. 
	It is decidable whether an equivalent linear transducer $T$  exists and if so,
	$T$ can be effectively constructed.
\end{theorem}

A transducer is \emph{partial} if $\text{rhs}$, seen as a function is partial.  
We remark that while it can be decided in polynomial time whether or not for a given total transducer
an equivalent linear transducer exists, 
the time complexity is double-exponential if the transducer is partial.
This is because the canonical earliest normal can be constructed with that 
complexity~\cite{DBLP:journals/jcss/EngelfrietMS09}; the same holds in the
presence of regular-look ahead.

Though for a given total transducer it is decidable in polynomial time whether or not 
an  equivalent linear transducer exists, the linear transducer of our construction  may be exponentially
larger than the given transducer.

\begin{example}
	Let $\Sigma=\{a^1,e^0\}$ and $\Delta=\{f^2, e,^0\}$.
	Let $T=(Q,\Sigma,\Delta, R, A)$ where $Q=\{q_0,\dots, q_9\}$, $A=q_0(x_1)$ and
	\begin{center}
		$\begin{array}{lcl c lcl }
		q_{i} (a(x_1)) \to f( q_{i+1} (x_1), q_{i+1} (x_1) )&\quad &   q_{i} (e) &\to&  e  \\
		q_{9} (a(x_1)) \to f( e, e )&\quad &   q_{9} (e) &\to&  e  \\
		\end{array}$
	\end{center}
where $i=0,\dots ,8$ be the rules of $R$. The transducer $T$ transforms a monadic tree $a^n e$ into a full binary tree of the same height if $n<10$ and into  a full binary tree of height 10 if $n\geq 10$. Clearly, $T$ is canonical earliest.
Furthermore, $T$ is zero output twined as there are no loops and lca-conform as its input trees are monadic.
We denote by $t_j$, $j=0,\dots ,9$, the full binary tree of height $j$ whose leaves are labeled
by $q_j$ and whose remaining nodes are labeled by $f$. 
The linear transducer our construction yields is  $T'=(Q',\Sigma,\Delta, R', A')$ where $Q'=\{\langle t_j \rangle \mid j=1,\dots, 9 \}$,
 $A'=\langle t_0 \rangle (x_1) =\langle q_0 \rangle (x_1)$ and 
	\begin{center}
	$\begin{array}{lcl c lcl }
	\langle t_{i}\rangle  (a(x_1)) \to \langle t_{i+1}\rangle  (x_1) &\quad &   \langle t_{i}\rangle (e) &\to&  t_i [q_i\leftarrow e]   \\
		\langle t_{9}\rangle  (a(x_1)) \to t_9 [q_9\leftarrow f(e,e)]&\quad &  \langle t_{9}\rangle (e) &\to&    t_9 [q_9\leftarrow e]  \\
	\end{array}$
\end{center}
where $i=0,\dots ,8$ are the rules of $R'$. Informally, $T'$ must delay its output until the leaf is read or until  
$T'$ has verified that its input tree has height at least 10 due to its linearity. The argument is analogous to 
Example~\ref{zo-twined introducrion example}.
Clearly, $T'$ is exponentially larger than $T$ as $\text{rhs}(\langle t_{9}\rangle, a)$ alone has
size $2^{10}$. 
\end{example}

%% file: homomorphism.tex
\section{When is a Transducer Equivalent to a Homomorphism?}

In this section we address the question whether or not a given (total) transducer is equivalent 
to some homomorphism. 
 In the following let
$T$ be a canonical earliest total transducer with the same tuple as defined in Section~\ref{transducer section}.
If we consider the task of deciding whether or not a   homomorphism $h$ equivalent to $T$ exists, then
it is tempting to believe that  $h$ is equivalent to $T$ iff the  canonical earliest normal form  of $T$ has
only one state. 
Interestingly, this is not the case.
\begin{example}\label{hom example 1}
	Let $\Sigma =\{a^1, e^0\}$ and $T_1$ be 
	defined with the axiom $f(q_1 (x_1), q_2 (x_1))$ and the
	rules
	\begin{center}
			$\begin{array}{lclclcl}
		q_1 (a(x_1))) &\to& f(q_1 (x_1), q_2 (x_1)) &\quad &q_2 (a(x_1))) &\to& f(q_1 (x_1), q_2 (x_1)) \\
		 q_1 (e)&\to & a  &\quad & q_2 (e) & \to & b  
		\end{array}$
	\end{center}
	Clearly $T_1$ is  canonical earliest and equivalent to the homomorphism
	$h$ with
		\begin{center}
			$\begin{array}{lclclcl}
	 h(a(x_1))&=&f(h(x_1),h(x_1)) &\quad  & h(e)&=&f(a,b)
\end{array}$
\end{center}

	This is rather surprising when considering \textit{string transducers}, i.e., transducers restricted to monadic trees.
	It is known that  canonical  earliest string transducers are state-minimal, i.e., there is no equivalent   string transducer
	with fewer states.~\cite[Section~2.1.2]{DBLP:phd/hal/Lhote18}\cite[Section~5]{DBLP:journals/tcs/Choffrut03} Hence, the decision problem for string transducer corresponds  to determining the equivalent
	canonical  earliest string transducer. 
\end{example}
 
 While the general case is more complicated than expected (see  Example~\ref{hom example 1}), the decision problem is simple
if the axiom of $T$ is of the form $q_0(x_1)$. 

\begin{lemma}\label{easy hom case}
Let $T$ be a    canonical earliest transducer with axiom
$q_0(x_1)$.  Whether or not an equivalent homomorphism $h$ exists is decidable in linear time.
\end{lemma}
\begin{proof}
	If  $T$ has only  one state, then $T$ is by definition a homomorphism.
	If, on the other hand, $T$ has more than one state, then there is no homomorphism equivalent to $T$. 
	Assume to the contrary that there is a homomorphism $h$ equivalent to $T$. 
	Due to the equivalence of $h$
	and $T$, $h (s)={T} (s)= \makro{q_0} (s)$  holds for all
	$s\in T_\Sigma$. As $T$ is earliest, there are 
	trees $s_1,s_2\in T_\Sigma$ such that $h(s_1)=\makro{q_0} (s_1)$ and $h(s_2)=\makro{q_0} (s_2)$ have different root symbols.
	Hence, $h$ is earliest too.
	As $h$ has only a single state it is obvious that $h$ is a canonical earliest transducer.
	The existence of $h$ contradicts the canonicity of $T$ according
	to Theorem~15 of~\cite{DBLP:journals/jcss/EngelfrietMS09}.
	Clearly the number of states of $T$ can be determined in linear time. \qed
\end{proof}

Subsequently, we therefore assume that the axiom of $T$ is not of the form $q_0 (x_1)$.
Furthermore, we assume that the axiom is not ground as this case is trivial.
Let $t_1 \in T_\Sigma [Q(X)]$ and $t_2\in T_\Sigma [Q(X)]$. Denote by $X_{t_1}$ and $X_{t_2}$ the elements of $X$
that occur in $t_1$ and $t_2$, respectively.
We write $t_1 \simeq t_2$  if there is a bijection  $f$ from $X_{t_1}$ to $X_{t_2}$, such that 
$t_2=t_1[q(a) \leftarrow q(f(a)) \mid q\in Q, a\in X_{t_1}]$.
We say that a tree $t\in T_\Sigma [Q(X)]$ is \emph{subtree conform} to $t'\in T_\Sigma [Q(X)]$
if either 
(1) $t$ is ground or,
(2)   $t \simeq t'$, or
(3) $t=a [t_1,\dots,t_k]$, $a\in \Sigma_k$,
and  $t_i$ is subtree conform to $t'$ for all $i\in [k]$.

\begin{example}
Consider the trees $t_1=f (q_1 (a_1), q_2 (a_2) )$,  $t_2=f (q_1 (b_1), q_2 (b_1) )$, and $t_3=f(q_1 (a_3), q_2 (a_3))$. Clearly $t_1 \not\simeq t_2$ as there is no suitable bijection, but $t_3\simeq t_2$ holds.
 The tree $t= d(e, g (f (q_1 (a_1), q_2 (a_1) )),
f (q_1 (a_2), q_2 (a_2)  ))$
is subtree conform to $t_2$. While $t\not \simeq t_2$, its subtree $e$  is ground,
its subtree $ g (f (q_1 (a_1), q_2 (a_1) )$  is subtree conform to $t_2$ and for its last subtree, $f (q_1 (x_2), q_2 (x_2)  )\simeq t_2$ holds. 
\end{example}

Let $T$ be an arbitrary  canonical earliest transducer with the axiom $A$.
In the following, we  show that 
$T$ is equivalent to a homomorphism $h$ if and only if  ${T} (a(x_1,\dots,x_k))$ is subtree conform to $A$ for all $a\in \Sigma$. 
Consider Example~\ref{hom example 1}.
The  transducer $T_1$ is equivalent to the homomorphism $h_1$ and
${T_1}(a(x_1))$ and ${T_1}(e)$ are  subtree conform to $f(q_1 (x_1), q_2 (x_1))$.

\begin{lemma} \label{hom lemma 1}
	Let $T$ be a canonical earliest transducer and $A$ be the axiom of $T$.
	If  $T$ is equivalent to a homomorphism $h$, then
	 ${T} (a(x_1,\dots,x_k))$ is subtree conform to $A$ for all $a\in \Sigma$.
\end{lemma}

\begin{lemma}\label{hom lemma 2}	
Let $T$ be a  transducer and $A$ be its axiom.
If 
${T} (a(x_1,\dots,x_k))$ is subtree conform to $A$ for all $a\in \Sigma$,
then $T$ is equivalent to a homomorphism.
\end{lemma}
\begin{proof}
	We construct the homomorphism $h$ as follows. Let $a\in \Sigma_k$, $k\geq 0$.
	If ${T} (a(x_1,\dots,x_k))$ is ground, then we define $h (a(x_1,\dots,x_k))= 	{T} (a(x_1,\dots,x_k))$.
	Otherwise,
	 our premise yields ${T} (a(x_1,\dots,x_k))=t[x_j\leftarrow A({j}) \mid j\in [k] ]$ 
	for a suitable $t\in T_\Delta [X_k]$ 
	where $A(j)=A[
	q(x_1)\leftarrow q(x_j) \mid q\in Q]$.
	Then we define  $h (a(x_1,\dots,x_k))= 
	t[ x_j\leftarrow h({x_j}) \mid j\in [k] ]$. We prove
	the equality of $h$ and $T$  in the Appendix. \qed
\end{proof}

A homomorphism equivalent to a given  transducer $T$ exists, 
if and only if ${T}(a(x_1,\dots,x_k))$ is subtree conform to $A$ for all $a\in \Sigma$
due to Lemmas~\ref{hom lemma 1} and~\ref{hom lemma 2}.
In the Appendix we show that subtree conformity to $A$ is decidable in polynomial time
and that an equivalent homomrphism  can be constructed in  polynomial time in the affirmative case.
Hence the following theorem holds.

\begin{theorem}\label{total hom decision}
	The question whether  a  canonical earliest transducer is equivalent 
	to a homomorphism is decidable in polynomial time. In the affirmative case,
	such a homomorphism can be constructed in  polynomial time.
\end{theorem}

\subsection*{Transducers with Regular Look-Ahead and  Homomorphisms}
We now consider partial
 transducers with regular look-ahead
  (or \emph{la-transducers} for short).
 Such devices consist of a partial transducer and a \emph{bottom-up tree automaton}, called the \emph{look-ahead automaton}
 (or \emph{la-automaton} for short).
 Informally, an la-transducer processes its input in two phases:
 First each input node is annotated by the active states of the bottom-up
 automaton at its children, i.e., an input node $v$ labeled by  $a\in \Sigma_k$ is relabeled by $\langle a,l_1,\dots , l_k \rangle$
 where  $l_i$ is the state the look ahead automaton  arrives in when processing the $i$-th subtree of $v$.
 Using the information the relabeled nodes provide about their subtrees, the transducer then processes the relabeled tree.
 
 The basic idea for answering the question whether there is an equivalent homomorphism for an arbitrary partial transducer $M$ with regular look-ahead
 is as follows.
 Let $\Sigma$ be the input alphabet of $M$.
  If $M$ is equivalent to a homomorphism, a subset $\Sigma'$ of $\Sigma$ exists such that
  $M$ is total if restricted to $T_{\Sigma'}$, i.e., $M(s)$ is defined for all $s\in T_{\Sigma'}$,  and undefined for all trees in $T_\Sigma\setminus T_{\Sigma'}$.
  Let $M'$ be $M$ restricted to $T_{\Sigma'}$. 
  Assume that a homomorphism $h$ equivalent to $M'$ exists, then
 there  exists a transducer $T$ equivalent to $M'$. (This is because $h$ is by definition  a  transducer.)
As $M'$ is obviously total, we apply the decision procedure in \cite{DBLP:journals/tcs/EngelfrietMS16} to determine
 whether $T$ exists and to construct the transducer $T$  as described in 
 \cite{DBLP:journals/tcs/EngelfrietMS16} --- if affirmative. 
 As $T$ is total,
 we can determine the homomorphism $h$ equivalent to $T$ (and hence $M'$) as described before.
 Clearly, a homomorphism equivalent to $M$ can be obtained from $h$.
For details on how to determine $M'$ and proofs we refer to the Appendix. 

In \cite{DBLP:journals/tcs/EngelfrietMS16}
it is shown that for every total la-transducer one can construct an equivalent canonical earliest la-transducer. 
Thus,  in the following assume that $M'$ is canonical earliest.
Note that 
in order to use the decision in \cite{DBLP:journals/tcs/EngelfrietMS16} to
construct a transducer  $T$ equivalent to $M'$
 we require a \textit{difference bound} for $M'$. A difference bound is a natural number
  $\text{dif}$ such that  $\text{dif}$ is an upper bound on the height of the \emph{difference trees} of $M'$. 
  Difference trees are defined as follows.
Consider some context $c$ over $\Sigma$ and states $l_1,l_2$ of the la-automaton of $M'$. 
Intuitively, when processing $cs$, $s\in T_\Sigma$, $T$ cannot make use of any annotations and hence  does not  know
whether $l_1$ or $l_2$
would be reached by
 the la-automaton when processing $s$.
Therefore, when processing $c$, $T$ can at most generate the largest common prefix
 $t$ of $M'(c l_1)$ and $M' (c l_2)$,
  where for $i=1,2$, $M' (c l_i)$ denotes output of $M'$ when processing $c$
 with the information that its la-automaton reaches $l_i$ when processing $s$.
 We formally prove this statement in the appendix.
 Let $v$ be a node such that $t[v]=x_1$, then $M' (c l_1)/v$ and $M' (c l_2) /v$
  are called  \textit{difference trees} of $M'$.
Intuitively, to simulate $M'$, $T$ has to store difference trees in its states,
which means that if $M'$ and $T$ are equivalent, then only finitely
many difference trees can exist. This in turn means that the height of all difference trees must be bounded.

In the following we show that the difference bound of $M'$ can be determined
if $M'$ is equivalent to a homomorphism $h$.


\begin{lemma}\label{difference tree bound}
	Let $l_1,\dots,l_n$ be states of the la-automaton and $s_1,\dots ,s_n$ be trees of minimal height that reach these states.
	Let $i\in[n]$ such that $\text{height}(M'(s_i))\geq \text{height}(M'(s_j))$, $j\neq i$, then, $\text{height}({M'} (s_i))$ is a difference bound 
	of $M'$. 
\end{lemma}

Using Lemma~\ref{difference tree bound} we can easily determine $\text{dif}\in \mathbb{N}$ such that $\text{dif}$ is an upper bound on the height of the difference trees of $M'$ if a homomorphism equivalent to $M'$ exists. Using $\text{dif}$ we can apply
 Algorithm 44 of  \cite{DBLP:journals/tcs/EngelfrietMS16} to determine whether there is a transducer equivalent to $M'$.
If Algorithm~44 of  \cite{DBLP:journals/tcs/EngelfrietMS16} yields ``no'' then obviously there cannot be a homomorphism equivalent to
$M'$. If on the other hand Algorithm~44 yields a transducer $T$, then we have to test whether $T$ and $M'$ are truly equivalent.
Equivalence of $T$ and $M'$ is decidable~\cite[Theorem 2]{DBLP:journals/ijfcs/Maneth15}. We must test this equivalence
because our input for Algorithm~44 of  \cite{DBLP:journals/tcs/EngelfrietMS16}  may yield \emph{false positives}. This occurs if 
there is no homomorphism equivalent to $M'$ as in this case the difference bound obtained via  Lemma~\ref{difference tree bound} is
potentially wrong which may cause Algorithm~44 to generate a transducer $T$ that is not equivalent to $M'$.
If our test yields that $M'$ and $T$ are not equivalent, then we output ``no''.

Otherwise we can determine whether there is a homomorphism $h$ equivalent to $T$ as described in the previous subsection.

 \begin{theorem}
 	It is decidable whether there is a homomorphism equivalent to a given la-transducer.
 \end{theorem}

%% file: conclusions.tex
\section{Conclusions}

We have proved that for a deterministic top-down
tree transducer $M$ with look-ahead, it is decidable
whether or not its translation can be realized by
a linear transducer (without look-ahead).
We have further shown that for such a transducer $M$ it is decidable whether
or not its translation is realizable by a one-state 
transducer (called a tree homomorphism).
In both cases, equivalent transducers in the respective subclass can
be constructed if affirmative.

One may wonder whether our results can be 
generalized to larger classes of given transducers.
It can easily be generalized to nondeterministic
top-down tree transducers: first decide if the transducer
is functional~\cite{DBLP:journals/actaC/Esik81} 
and if so, construct an equivalent
deterministic top-down tree transducer with 
look-ahead~\cite{DBLP:journals/ipl/Engelfriet78}.
Note that the result of Engelfriet~\cite{DBLP:journals/ipl/Engelfriet78} shows that
for any composition of nondeterministic top-down and 
bottom-up transducers that is functional, an equivalent
deterministic top-down tree transducer with look-ahead
can be constructed.
This raises the question, whether or not for a composition
of nondeterministic transducers, functionality is
decidable. To the best of our knowledge, this is an 
open problem.

In future work, it would be nice to extend our result of 
deciding homomorphisms within deterministic
top-down tree transducers, to the case that for a
given $k$ one decides whether an equivalent 
top-down tree transducer with $k$ states exists
(and if so construct such a transducer). 
This would offer a state-minimization method.

%% file: look-ahead.tex
\section{Transducers with Look-Ahead}\label{la-transducer section}
A \emph{transducer with regular look-ahead} (called \emph{la-transducer} for short) $M$
is a tuple $(Q,\Sigma,\Delta,R,B,A)$, where $Q,\Sigma$ and $\Delta$ are as for transducers, 
$B$ is the \emph{look-ahead automaton} (\emph{la-automaton}), and
$A$ assigns an axiom to each accepting state of $B$.
Here, the la-automaton $B$ is a deterministic bottom-up automaton, i.e., a tuple
$B=(L,\Sigma,\delta,F)$, where $L$ is a finite set of states,
$F\subseteq L$ is the set of final states,
$\Sigma$ is the input alphabet, and
$\delta$ is the (possibly partial) transition function.
The states of  $B$ are called \emph{la-states}.
The transition function $\delta$ is defined as
$\delta(a,l_1,\dots,l_k)=l$, where $a\in \Sigma_k$ and $l,l_1,\dots,l_k\in L$.
The mapping $\delta$ can be extended to a function
$\delta^*$ from $T_\Sigma [L]$ to $L$ in the obvious way.
We denote by $\text{dom}(l)$ the set of trees $t$ with $\delta^*(t)=l$
and let $\mathcal{L}(B)=\bigcup_{l\in F}\text{dom}(l)$. 
If $L=F$ then $B$ is \emph{total}.

For all $q\in Q$, $a\in \Sigma_k$ and $l_1,\dots, l_k\in L$
the set $R$ contains at most one rule of the form $q(a(\stat{x_1}{l_1},\dots, \stat{x_k}{l_k})) \to t$ where $t\in T_\Delta [Q(X_k)]$.
The right-hand side of such a rule is denoted by $\text{rhs}_M (q,a,l_1,\dots,l_k)$.
If $M$ is  clear from the context, then we omit the subscript $M$ and write $\rhs{q,a,l_1,\dots,l_k}$ instead.

To define the semantics of $M$ we consider the la-automaton $B$
as a (partial) function  annotating each input node by an appropriate sequence of states in $L$.
We demand that $B$  annotates symbols in $L$  by $x_1$. 
For instance, consider $B_a=(\{p,p'\},\{f^2,a^0,b^0\},\delta,\{p'\})$
with $\delta(a)=p'$, $\delta(b)=p$ and $\delta(f,p',p)=p'$, 
$\delta(f,p,p')=p$, $\delta(p',p')=p'$, and
$\delta(f,p,p)=p$. It accepts trees whose left-most leaf is labeled $a$.
For  trees $t_1=f(b,f(a,b))$  and $t_2=f(b,f(p,p'))$ we have
\[
B_a(t_1) = \langle f,p,p'\rangle(b,\langle f,p',p\rangle(a,b)) \text{ and } B_a(t_2) = \langle f,p,p'\rangle(b,\langle f,p',p\rangle(x_1, x_1)). 
\]
For $s\in T_\Sigma [L]$, 
the semantics of $M$ is 
 $M(s) = \makro{M}(s) = A(l) [q(x_1)\leftarrow \makro{q} (B(s)) \mid q\in Q]$
where a rule $q(a(\stat{x_1}{l_1},\dots, \stat{x_k}{l_k})) \to t$ is interpreted to be of the form
$q(\langle a,l_1,\dots,l_k\rangle(x_1,\dots,x_k))\to t$. For la-transducers we modify the semantic of
$\makro{q}$ slightly.
Let $V$ be the set of nodes.
For $q\in Q$, $u\in V$, we denote by $\makro{q}_u$  the partial function from
$T_\Sigma [L]$ to $T_\Delta [Q(V)]$ defined  as follows:
\begin{itemize}
	\item 
	$\makro{q}_u(s)=\rhs{q,a}[q(x_i) \leftarrow \makro{q}_{u.i} (s_i) \mid q\in Q, i\in [k]]$
	for $s=a(s_1,…,s_k)$, $a\in \Sigma_k$, and $s_1,\dots,s_k\in T_\Sigma[L]$ 
	\item $\makro{q}_u(l)= q(u)$ for $l\in L$.
\end{itemize}
The reason why (unconventionally) the input nodes $u$ are added
to the semantics will become apparent later. We write $\makro{q}$ instead of $\makro{q}_u$  to refer to an arbitrary $u$.
For la-transducers we define \emph{linearity} analogously to ordinary transducers.
The la-transducer $M$ is \emph{total} if the domain of the function $\makro{M}$ equals $T_\Sigma$.
In the this case, the la-automaton of $M$ is total as well.
Two la-transducers $M_1$ and $M_2$ are \emph{equivalent} if
the functions $\makro{M_1},\makro{M_2}$ restricted to ground input trees
are equal.

\begin{example}
Consider the la-transducer $M_a=(Q_a,\Sigma,\Delta, R_a,B_a,A_a)$, where
$B_a$ is the la-automaton of before, but with set $F$ changed to $\{p,p'\}$.
We want to perform the identity on input trees with left-most leaf $a$,
and for all other input trees want to change each label to a primed version.
Thus,
\[
A_a(p)=q(x_1)\text{ and }A_a(p')=q_{\text{id}}(x_1).
\]
Further for all $l_1,l_2\in L$ the la-transducer $M_a$ has the rules:
\[
\begin{array}{lcl}
q(f(\stat{x_1}{l_1},\stat{x_2}{l_2}))&\to& f'(q(x_1),q(x_2))\\
q_{\text{id}}(f(\stat{x_1}{l_1},\stat{x_2}{l_2}))&\to& f(q_{\text{id}}(x_1),q_{\text{id}}(x_2)).
\end{array}
\]
plus for $d\in \{a,b\}$ we have rules
$q(d)\to d'$ and $q_{\text{id}}(d)\to d$.
\end{example}
%
An la-transducer $M$  is \emph{look-ahead uniform} (\emph{la-uniform}) if a mapping $\rho$ (called \emph{la-map} of $M$)
from $Q$ to $L$ exists such that
\begin{enumerate}[label=\textnormal{(\arabic*)}]
	\item if $q(x_1)$ occurs in $A(l)$ where $l\in F$ then $\rho(q)=l$ \label{itm:1}
	\item  $R$ contains a rule  $q(a(\stat{x_1}{l_1},\dots, \stat{x_k}{l_k})) \to t$ if and only if $\rho (q)=\delta (a,l_1,\dots l_k)$. Additionally, $\rho (q')=l_i$ if $q'(x_i)$ is contained in $\rhs{q,a,l_1,\dots,l_k}$.
\end{enumerate}

 \begin{lemma}\label{appendix 1}
	For every  la-transducer $M$ an equivalent  la-uniform la-transducer $N$  with the same la-automaton can be constructed.
\end{lemma}

\begin{proof}
	Let $M=(Q,\Sigma,\Delta,R,B,A)$. First, we construct a total la-transducer $M_t=(Q,\Sigma,\Delta_t,R_t,B_t,A_t)$ such that $M$ and $M_t$
	generate the same output when processing a tree $s \in \text{dom} (M)$. More precisely, $M_t$ is obtained form $M$
	as follows:
	Let $\bot$ be a fresh symbol of rank $0$ and $\square$ be a fresh la-state. The output alphabet of $M_t$ is $\Delta \cup \{\bot\}$ and
	the set of la-states of $M_t$ is $L \cup \{\square\}$, where $\Delta$ and $L$ are the output alphabet and set of la-states of $M$,
	respectively.
	We obtain the transition function $\delta_t$ of  $B_t$  by  expanding the transition function $\delta$ of $B$   as follows: if $\delta (a,l_1,\dots, l_k )$
	where $a\in \Sigma_k$ and $l_1,\dots,l_k$ are some la-states of $B$ is undefined, 
	then $\delta_t (a,l_1,\dots, l_k )=\square$ otherwise $\delta_t (a,l_1,\dots, l_k )=\delta (a,l_1,\dots, l_k )$. Additionally, if for some $i\in [k]$, $l'_i=\square$, then we define $\delta_t (a,l'_1,\dots, l'_k )=\square$. We define $A_t(l)=A(l)$ if $l$ is a la-state of $B$ and $A_t(\square) =\bot$.
	We obtain $R_t$ by expanding $R$ with rules of the form $q(a(\stat{x_1}{l_1},\dots, \stat{x_k}{l_k})) \to \bot$    if $\text{rhs}_M (q,a,l_1,\dots,l_k)$ is undefined
	where  $q$ is some state, $l_1,\dots,l_k$ are la-states and $a\in \Sigma_k$ or $l_j=\square$, $j\in[k]$.
	
	Clearly $M_t$ is total and  $M(s)=M_t(s)$ for all   $s \in \text{dom} (M)$. Furthermore, it easy to see that if $s' \notin \text{dom} (M)$,
	then the tree $M_t(s')$ contains occurrences of $\bot$. Thus, we conclude that $s \in \text{dom} (M)$ if and only if $M_t (s) \in T_\Delta$
	holds $(*)$.
	By Lemma~13 of \cite{DBLP:journals/tcs/EngelfrietMS16}, we can construct an la-uniform la-transducer $M'_t=(T'_t, B'_t, A'_t)$
	that is equivalent to $M_t$
	and has the same la-automaton $B_t$ as $M_t$. Let $\rho$ be the corresponding la-map.
	For $M'_t$ the following holds
	\begin{description}
		\item[(a)]  For all $s\in \text{dom} (M)$ it holds that no axiom or rule containing $\bot$
		occurs in the computation of $M'_t(s)$ due to $(*)$ as $M_t'$ and $M_t$ are equivalent. Furthermore,
		if $s\notin \text{dom} (M)$ then some axiom or rule containing $\bot$
		occurs in the computation of $M'_t(s)$.
		\item[(b)]  If $s\in \text{dom} (M)$, then 
		for all subtrees $s_\text{sub}$ of $s$, $\delta_t (s_\text{sub})\neq \square$,
		otherwise $\delta_t^*(s)=\square$ and  $M_t'(s)=M_t(s)=\bot$ contradicting $(*)$. Thus no state $q$ such that
		$\rho (q)=\square$ occurs in the computation $M'_t(s)$.
	\end{description}
	Let $N$ be the la-transducer that is obtained by removing all axioms and rules that involve the symbol $\bot$, the la-state $\square$,
	or states $q$ such that $\rho (q)=\square$ from $M_t$. Clearly, the la-automaton of $N$ is $B$.
	We remark that removing these rules and axioms  preserves the la-uniformity.
	Due to (a) and (b) it follows that $N(s)=M_t'(s)=M(s)$ if $s\in \text{dom} (M)$.
	In particular, $N(s)$ is undefined if $s\notin \text{dom}(M)$ due to (a) as $N$ cannot produce any tree in which the symbol $\bot$
	occurs. Thus, due to (b) the la-state $\square$ is not needed. Hence $B$ is sufficient as a la-automaton.\qed
\end{proof}

Note that if $M$ is la-uniform then  all states $q$ such that $\bigsqcup \{\makro{q}(s)\mid s\in \text{dom}(q) \}\neq x_1$
can be determined as in \cite{DBLP:journals/tcs/EngelfrietMS16}.

We say that an la-uniform la-transducer $M$ is \emph{earliest} if $\bigsqcup \{\makro{q}(s)\mid s\in \text{dom}(q) \}=x_1$ holds for all $q\in Q$. 
Informally this means that
for each state $q\in Q$ there are input trees $s_1,s_2 \in T_\Sigma$ such that $\makro{q}(s_1)$ and $\makro{q} (s_2)$ have different root symbols.
We call $M$ \emph{canonical earliest} if $M$ is earliest and for distinct states $q_1$, $q_2$ it holds that if $\rho(q_1)=\rho(q_2)$ then a ground tree $g$ exists such that $\makro{q_1}(g)\neq \makro{q_2}(g)$.

\begin{theorem}\label{canonical earliest theorem}
	For every  la-transducer $M$ an equivalent  canonical earliest la-transducer $N$  with the same la-automaton can be constructed.
\end{theorem}

\begin{proof}
	The proof is identical to the proof of Theorem~15 in \cite{DBLP:journals/tcs/EngelfrietMS16} (even though Theorem~15 is stated for total
	la-transducer).\qed
\end{proof}

%% file: appendix.tex
\section{Omitted Proofs}
In the following we  present the proofs that have previously been omitted.
Note that although the following proofs are stated for la-transducers, they 
are applicable to total transducers without look-ahead as well, as 
any such transducer is essentially a total la-transducers whose la-automaton realizes the identity.

Let  $t_1,\dots,t_k$ be trees and $t\in T_\Sigma [X_k]$ with $k> 0$.  
In the following instead of writing $t[x_i\leftarrow t_i \mid i\in[k]]$ we also write $t[t_1,\dots, t_k]$.
  In the following let
$M$ be an la-transducer (with the same tuples) as defined in Section~\ref{la-transducer section}
that is canonical earliest.
 
 \subsection{Proof of Lemma~\ref{auxiliary total lemma}}
 Intuitively, any la-transducer $N$ that
 is equivalent to a canonical earliest la-transducer $M$ and has the same la-automaton as $M$ cannot 
 generate ``more'' output than $M$ on the same input.
 Therefore the following holds.
 
 \begin{lemma}\label{auxilliary appendix}
 	Let $M$ and $N$ be equivalent la-transducers with the same la-automaton. Let $M$ be  canonical earliest.
 	Let $s\in T_\Sigma [L]$ such that $\delta^* (s)\in F$.  Then  $V(N(s)) \subseteq V(M(s))$ and  if $N(s)[v]=d \in \Delta$ then 
 	$M(s)[v]=d$ for all $v\in V(N(s))$.
 \end{lemma}
 
 \begin{proof} 
 	Let $B=(L,\Sigma,\delta,F)$ be the la-automaton of $M$ and $N$.
 	Assume to the contrary that a node $v\in V(N(s)) \setminus V(M(s))$ exists.	
 	As $M$ and $N$ are equivalent, there must exist a proper
 	ancestor $v'$ of $v$ such that $v'$ is labeled by $q(u)$ in  $M(s)$ where $q$ is some state in $Q$ where $u\in V(s)$ such that $s[u]=l\in L$. (This implies $\rho (q)=l$.) Obviously, $v'$ is labeled by some symbol $d\in\Delta$ in  $N(s)$. We show that no such node $v'$
 	can exist.
 	Clearly $s'=s [l'\leftarrow s_{l'} \mid l'\in L, s_{l'}\in \text{dom} (l')]\in \text{dom} (B)$ and therefore $s'\in \text{dom} (M)$ as $M$ is
 	la-uniform. 
 	Hence, $M (s')[v']=N (s') [v']$  as $M$ and $N$ are  equivalent.  Therefore,   $\makro{q}(s_l) [\lambda]=d$  for all
 	$s_l\in \text{dom} (l)$. This contradicts the the earliest property of  $M$ as $\text{dom}(q)=\text{dom} (l)$
 	due to the la-uniformity of $M$.
 	
 	We have shown above that $V(N(s)) \subseteq V(M(s))$ and that no node $v$ such that $M(s) [v] \notin \Delta$ and  $N(s) [v] \in \Delta$ exists.
 	Hence, if $N(s)[v]\in \Delta$ then $M(s)[v]\in \Delta$ for all $v\in V(N(s))$. Therefore,
 	the equivalence of $M$ and $N$ yields that   if $N(s)[v]=d \in \Delta$ then 
 	$M(s)[v]=d$ for all $v\in V(N(s))$. \qed
 \end{proof}

Let $M$ and $N$ be equivalent la-transducers with the same la-automaton. Let $M$ be  canonical earliest.
Informally, whenever $M$ is processing some input, $N$ also
needs to process that input (but possibly at an output ancestor node, if here
$N$ lags behind $M$).
Due to Lemma~\ref{auxilliary appendix}, the trees generated 
by $M$ and $N$ on input $s$ start to ``diverge'' at the nodes in $ V_{\neg\Delta} (N(s))$, where $ V_{\neg\Delta} (N(s))$ denotes 
the set of all nodes in $V(N(s))$ not labeled by symbols in $\Delta$. Thus, the following holds.

\begin{corollary}\label{auxilliary appendix 2}
	Let $M$ and $N$ be as in Lemma~\ref{auxilliary appendix}.
	Let $s\in T_\Sigma [L]$ such that $\delta^* (s)\in F$.  Then $N(s) [ v\leftarrow  M(s)/v \mid v \in V_{\neg\Delta} (N(s)) ]=   M(s) $.
\end{corollary}

 \subsection{Proof of Lemma~\ref{zero output twinned}}
 For la-transducers \emph{zero output twinnedness} is defined analogously to ordinary transducers.
 Before we prove that an la-transducer $M$ that is equivalent to a la-transducer $N$ that has the same la-automaton as $M$
 must be zero output twinned, we prove the following auxiliary results.
 
  \begin{proposition}
 \cite[Lemma~12]{DBLP:journals/tcs/EngelfrietMS16}
 Let $s\in T_\Sigma [L]$, $M$ be an la-transducer and $q$ be a state of $M$.
 Note that   if  $q(u)$ occurs in  $\makro{q'} (s)$, where $q'$ is some state,  then  $\rho (q)=s[u]$ where $s[u]\in L$.
\end{proposition}

If  $q(u)$ occurs in  $M(s)$, then some state $q'$ must exist such that $q'(x_1)$ occurs in an axiom and 
$q(u)$ occurs in  $\makro{q'} (s)$. Hence the following holds.
 
  \begin{corollary}\label{appendix corollary}
	Let $s\in T_\Sigma [L]$, $M$ be an la-transducer.
	Note that   if  $q(u)$ occurs in  $M(s)$,  then  $\rho (q)=s[u]$ where $s[u]\in L$.
\end{corollary}

 \begin{proposition}
 	Let  $q_1$ and $q_2$ pairwise occurring states in $M$, $c$ be a context and $l$ be a la-state. Let 
 	$(q_1,q_2) {\vdash}^{c l} (q_1, q_2)$. Then  $l=\rho (q_1)=\rho (q_2)$
 	and $\delta^*(c l)=l$.
 \end{proposition}
 \begin{proof}
 	If $q_i (u)$ occur in  $\makro{q_i} (c l)$, $i=1,2$ then, due to the previous proposition, $c l[u] =\rho (q_i)$.
    Note that as only a single node in $c l$ is labeled by an la-state $c l[u]=l$. Thus  $l=\rho (q_1)=\rho (q_2)$.
 	
 As $M$ is la-uniform, it holds for an arbitrary state $q$ and an arbitrary tree $s$ that $\makro{q} (s)$ is defined if
 and only if $\delta^* (s)=\rho (q)$ where  $\bar{s}\in T_\Sigma [L]$.
 As $\makro{q_1} (c_l)$ and $\makro{q_2} (c_l)$ are defined and   $l=\rho (q_1)=\rho (q_2)$, our second claim follows.\qed
 \end{proof}

 \noindent 
Now  we prove the first part of our main claim.
 
 \begin{lemma}
 	If $M$ is equivalent to a linear la-transducer $N$ that has the same la-automaton as $M$, 
 	then $M$ is zero output twinned. 
 \end{lemma}
 
 \begin{proof}
 	We prove  by contradiction that  $M$ is   zero output twinned if $M$ is equivalent to a linear la-transducer $N$ with the same la-automaton.
 	Let $c$ be a context and $l$ be a la-state such that distinct nodes $v_1$ and $v_2$ are labeled $q_1 (u)$ and $q_2 (u)$, respectively, in $M(cl)$. Clearly , $q_1$
 	and $q_2$ are pairwise occurring states. Note that only a single node  labeled by an la-state occurs in $c l$ and thus
 	$l=\rho (q_1)=\rho (q_2)$ due to previous corollary.
 	Let $c'$ be a context and $l'$ be a la-state such that $(q_1,q_2) {\vdash}^{c' l'} (q_1, q_2)$.
 	The  previous proposition  yields $\rho (q_1)=\rho (q_2)=l'$
 	and $\delta^*(c' l')=l'$.  Hence, $l=l'$.
 	Therefore, $M(c c'^{n} l)$ and $N(c c'^{n} l)$ are defined for all $n\in \mathbb{N}$.
 	Furthermore, it is clear that $M(c c'^{n} l)/v_i=\makro{q_i} (c'^{n} l)$, $i=1,2$. 
 	
 	W.l.o.g. assume that $q_1$ generates some output on input $c' l$.
 	Let $s\in T_\Sigma$ such that $\delta^* (s)=l$.
 	Then, it holds that $n'\in \mathbb{N}$ exists such that 
 	$\text{height} (\makro{q_1}({c}^{n'} s)) >  (\text{height}(s) \text{maxrhs}_N ),$ where $\text{maxrhs}_N$ is the maximal height of any right-hand side of $N$ and $s\in T_\Sigma$ such that $\delta^*(s)=l$.	 
 	
 	The tree ${M}(c c'^{n'} l)$ is not ground as $M(c c'^{n} l)/v_i=\makro{q_i} (c'^{n} l)$ is not ground for $i=1,2$.
 	Hence,  ${N}(c c'^{n'} l)$ is not ground due to Lemma~\ref{auxilliary appendix}
 	and due to the linearity of $N$ we conclude that there is exactly one node $v$
 	of  ${N}(c c'^{n'} l)$ not labeled by a symbol in $\Delta$. 
 	We  show that $v$ is an ancestor of $v_1$ and $v_2$. Note that as 
 	$v_1$ and $v_2$ are leaves  in ${M}(c l)$, $v_1$ and $v_2$ are disjoint.
 	Assume the contrary, i.e., that $v$ is either not an ancestor of $v_1$ or $v_2$, then
 	either  ${N}(c c'^{n'} l)/v_1$ or ${N}(c c'^{n'} l)/v_2$ must be defined and be ground, otherwise
 	 ${M}(c c'^{n'} s)=  {N}(c c'^{n'} s)$ cannot hold. Note that ${N}(c  {c'}^{n'}  s)={M}(c {c'}^{n'} s)$ as $M$ and $N$ are equivalent.
 	 This contradicts  
 	Lemma~\ref{auxilliary appendix} as both ${M}(c c'^{n'} l)/v_1$ and ${M}(c c'^{n'} l)/v_2$ are not ground.
 	
 	Let  ${N}(c c'^{n'} l)/v=q'(u')$ where $u'$ is the node of $c c'^{n'} l$ labeled $l$. 
 	Clearly, $q'$ can on input $s$ generate a tree of height at most $(\text{height}(s) \text{maxrhs}_N )< \text{height}({M}(c {c'}^{n'}  s)/v_1)= \text{height}\makro{q_1}({c'}^{n'} s)$. As
 	$v$ is an ancestor of $v_1$,
 	this contradicts the
 	equality ${N}(c  {c'}^{n'}  s)={M}(c {c'}^{n'} s)$. \qed
 \end{proof}

In the following we show how to determine whether or not $M$ is zero output twinned.
For this we construct  a set of pairs $\text{PO}$ such that  $(q_1,q_2)\in \text{PO}$ if and only if $q_1$ and $q_2$ are
pairwise occurring states and a
directed graph $G=(V,E)$.
We define $V:=\text{PO}\cup \{[q_1,q_2] \mid (q_1,q_2)\in  \text{PO}\}$. 
We define the set of edges $E$ as follows. Let $a\in \Sigma_k$, $k>0$, $l_1,\dots,l_k$ be la-states
and $(q_1,q_2)\in \text{PO}$ such that $\bar{q_i}(x_j)$ occurs in $t_i=\rhs{q_i,a,l_1,\dots,l_k}$, $i=1,2$, $j\in [k]$,  then
\begin{enumerate}
	\item  
	there is an edge
	from $[q_1,q_2]$ to $[\bar{q}_1, \bar{q}_2]$.
	\item   there is an edge
	from $(q_1,q_2)$ to $(\bar{q}_1, \bar{q}_2)$ if  $t_i=\bar{q_i}(x_j)$ for  $i=1,2$ 
	\item   there is an edge
	from $(q_1,q_2)$ to $[\bar{q}_1, \bar{q}_2]$ if  $t_1\neq \bar{q_1}(x_j)$ or $t_2\neq \bar{q_2}(x_j)$.
\end{enumerate}
Due to the la-uniformity of $M$, $\text{dom} (q)=\text{dom} (\rho (q))$. With this, structural induction yields that there is a path from
$(q_1,q_2)$ to  $[ \bar{q}_1,\bar{q}_2 ]$ if and only if there exits a context
$c$  and a look ahead state $l$  such that $(q_1,q_2) {\vdash}^{c  l} (\bar{q}_1, \bar{q}_2)$.

For the construction of $\text{PO}$ and hence $G$ consider the following.
If $q_1,q_2$ are pairwise occurring states then one of the following statements must hold
\begin{description}
	\item[(a)] distinct nodes $v_1$, $v_2$ with labels $q_1(x_1)$ and $q_2 (x_1)$ occur in $A(l)$ for some $l\in F$,
	\item[(b)] distinct nodes $v_1$, $v_2$ with labels $q_1(x_i)$ and $q_2 (x_i)$ occur in $\rhs{q,a,l_1,\dots,l_k}$ where $q\in Q$, $a\in \Sigma_k$, $l_1,\dots,l_k \in L$, and $i\in [k]$,
	\item[(c)]  $q_1(x_i)$ occurs in  $\rhs{q'_1,a,l_1,\dots,l_k}$  and
	$q_2 (x_i)$ occurs in $\rhs{q'_2,a,l_1,\dots,l_k}$   where $q'_1$, $q'_2$ are pairwise occurring states, $a\in \Sigma_k$, $l_1,\dots,l_k \in L$, and $i\in [k]$.
\end{description} 
All pairwise occurring states  such that $(a)$ holds can be identified
by comparing the labels of $v_1$ and $v_2$
for all pairs of distinct leaves  $(v_1, v_2)$ of $A(l)$ for all $l\in F$. Obviously $A(l)$ has $O(|M|)$ leaves and $|F| \in O(|M|)$. Thus this step is completed in time $O(|M|^3)$ and we add all identified pairs to $\text{PO}$.
All pairwise occurring states  $q_1,q_2$ such that $(b)$ holds can be identified as follows.
For all right hand sides $t$ of rules in $R$  
we compare the labels of  all pairs of leaves  $(v_1, v_2)$ of $t$. This yields all pairwise occurring states  such that $(b)$ holds. 
This step is completed in time $O(|M|^3)$ and again we add all pairs identified in this step to $\text{PO}$.
Starting from the set $\text{PO}$ we have constructed so far we identify all pairwise occurring states  $q_1,q_2$ such that $(c)$ holds 
inductively. Additionally we construct the set $E$ in this step.

Given  $(\bar{q_1}, \bar{q_2})\in \text{PO}$ we determine all pairs of right hand sides $t_1$, $t_2$ of $R$ such that $t_1=\rhs{\bar{q_1},a,l_1,\dots,l_k}$
and $t_2=\rhs{\bar{q_2},a,l_1,\dots,l_k}$ for some $a\in \Sigma_k$, $l_1,\dots,l_k \in L$. This can be done in  $O(|M|^2)$.
For $t_1$ and $t_2$ we  determine all a pair $(q_1,q_2)$ such that  
$q_1(x_i)$ and $q_2 (x_i)$ occur in $t_1$ and $t_2$, respectively, where $i\in [k]$. Determining all such pairs
can be done in time  $O( |M|^2)$.
We add all these pairs $(q_1,q_2)$ to $\text{PO}$. Additionally, we add an edge from $[\bar{q_1}, \bar{q_2}]$ to  $[q_1,q_2]$ to $E$.
Furthermore, if $t_1=q_1 (x_i)$ and $t_2= q_2 (x_i)$, then we
add an edge from $(\bar{q_1}, \bar{q_2})$ to $({q}_1,{q}_2)$  to $E$, otherwise we
add an edge from $(\bar{q_1}, \bar{q_2})$ to $[{q}_1,{q}_2]$.
Applying this procedure to all elements in $PO$ yields all pairwise occurring states  for which $(c)$ holds. 
Clearly, this step is completed in time $O( |M|^6)$ as the cardinality of $\text{PO}$
cannot exceed $|Q|^2 $.

It can be verified by structural induction that $\text{PO}$  contains all pairs of states that occur pairwise and
that the set $E$ has the properties we specified.
Clearly $\text{PO}$ and $G$ can be constructed in time $O( |M|^6)$.

The size of $G$ is $|V|+|E|\leq 2 |Q|^2+ 4 |Q|^4$, i.e., $|G|\in O( |Q|^4)$.
Determining whether $M$ is zero output twinned consists of applying depth-first search on $G$.
More precisely for each $(q_1,q_2)\in \text{PO}$ it needs to be determined whether $[q_1,q_2]$ is reachable from $(q_1,q_2)$.
This can be done in time $O(|Q|^6)$ as $|\text{PO}|\in O(|Q|^2)$ and $|G|\in O( |Q|^4)$. If $[q_1,q_2]$ is reachable, then clearly $M$ is not zero output twinned.

\subsection{Proof of Lemma~~\ref{lca-conform}}

For la-transducers we modify our \emph{lca-conformity} definition as follows.
Let $s\in T_\Sigma [L]$ such that  $\delta^*(s)\in F$.
We denote by $\nu_M (s,u)$ the lowest common ancestor of all leaves of $M(s)$ that have
label of the form $q(u)$ where $q\in Q$.
We say that $M$ is  \emph{lowest common ancestor conform} (\emph{lca-conform} for short) if  for an arbitrary (partial) input tree $s$ such that  $\delta^*(s)\in F$
the output  tree ${M}(s)$ is lca-conform. 
An output tree  ${M}(s)$  is \emph{lca-conform} 
if for all $u$ such that $\nu_M (s,u)$ is defined, no leaf with label of the form $q(u')$, $u\neq u'$, occurs in $M(s)/ \nu_M (s,u)$.

\begin{lemma}
	If $M$ is equivalent to a linear la-transducer $N$ that has the same la-automaton as $M$, then $M$ is  lca-conform.
\end{lemma}

\begin{proof}
	We prove  that  $M$ is  lca-conform if $M$ is equivalent to a linear la-transducer $N$ with the same la-automaton.
	
	Let  $N=(T',B,A)$ and $s\in T_\Sigma [L]$ such that $\delta^*(s)\in F$.
	Let $v$ be a node such that ${N}(s)/v= q'(u')$ where $q'\in Q'$  and $s[u']\in L$. 
	First, we show that the subtree ${M}(s)/v$  contains no leaves that have label of the form $q(u)$ with $u\neq u'$.	
	Note that due to Lemma~\ref{auxilliary appendix},  ${M} (s) / v$ is  defined.
	
	Assume that  ${M}(s)/v$ contains  some node labeled $q(u)$, i.e., ${M}(s)/ \hat{v}=q(u)$ for some descendant
	$\hat{v}$ of $v$.
	Consider the trees ${N}(s[u'\leftarrow g])$ and ${M}(s[u'\leftarrow g])$ where $g\in T_\Sigma$ such that $\delta^*(g)=s[u']$. Clearly   ${M}(s[u'\leftarrow g])/\hat{v}=q(u)$ and hence ${M}(s[u'\leftarrow g])/v$ is not ground. However ${N}(s[u'\leftarrow g])/v= \makro{q'} (g)$ is ground. This contradicts Lemma~\ref{auxilliary appendix}.
	Thus we deduce that for all $v\in V(N(s))$ if  $N(s) / v =q'(u')$ then no leaf with label of the form
	$q(u)$ with $u\neq u'$ and $q\in Q$ occurs in
	 $M(s)/v$ (*).
	
	We now show that if leaves with label of the form $q(u')$ occur in $M(s)$ then $v$ is an ancestor node of all those 
	leaves. Assume to the contrary that some node $\tilde{v}$ with label of the form $q(u')$ occur in $M(s)$ that is not a descendant of $v$ occurs in $s$.
	Due to 	Corollary~\ref{auxilliary appendix 2}, it follows that some ancestor of  $\tilde{v}$ is labeled by
	a symbol in $Q(L)$ in $N(s)$.  As  $\tilde{v}$ is not a descendant of $v$ and $N$ is linear, that symbol cannot be of the form $q''(u')$ where $q''$ is some state. This contradicts (*).
	Thus,  
	 we conclude that 
	if  $N(s) / v =q'(u')$ 
	and  leaves with label
	$q(u')$ occur in $M(s)$, then  $M(s)/v$ contains all these leaves and hence $v$  is a common ancestor of all these leaves.
	
	As ${M}(s)/v$  does not contain leaves that have label of the form $q(u)$ with $u\neq u'$  neither does $M(s)/ \nu_M (s,u)$ as $v$ is an ancestor of $\nu_M (s,u)$.
	Hence, ${M}(s)$ is la-conform. \qed
\end{proof}
\noindent
In the following we show how to decide whether or not a given la-transducer $M$ is lca-conform.
Before we begin consider the following auxiliary result:
In the following lemma  we show that a single letter suffices to turn a lca-conform tree to a tree that is not
lca-conform. Note that ${M}(c  l)$, where $c$ is a context over $\Sigma$ and $l\in L$ is obviously lca-conform.

\begin{lemma}\label{decision aux}
	If $M$ is not lca-conform, then a context $c$, $l_1,\dots,l_k \in L$, and $a\in \Sigma_k$, $k\geq 2$, exists such that 
	${M}(c  a[l_1,\dots,l_k])$ is not  lca-conform.
\end{lemma}

\begin{proof}
	If $M$ is not lca-conform, then  $s\in T_\Sigma [L]$ with  $\delta^*(s)\in F$  exists such that 
	only two distinct leaves $u_1, u_2 \in V(s)$ are labeled by a symbol not in $\Delta$, 
	and $M(s)$ is not lca-conform. Without restriction assume that
	some node with label $q (u_2)$ occurs in
	$M(s) / \nu_M (s,u_1) $, i.e., some descendant $v$ of $\nu_M (s,u_1)$ exists such that
	$M(s)/ v = q (u_2)$. Obviously, all ancestors of $v$ are either ancestors or descendants of $\nu_M (s,u_1)$.
	Hence,  $\nu_M (s,u_1)$ and $\nu_M (s,u_2)$ cannot be disjoint (*).
	
	Let $u$ be the lowest common ancestor of $u_1$ and $u_2$. Without restriction let $u.j$ be an ancestor of $u_j$, $j=1,2$.
	Consider the tree $M(s')$ where $s'=s [u\leftarrow a(l_1,\dots,l_k) ]$ with  $s[u] = a\in \Sigma _k$ and $\delta^* (s/u.i )=l_i$, $i\in [k]$.
	Note that clearly a context $c$ exists such that $
	c  a(l_1,\dots,l_k) =
	s [u\leftarrow a(l_1,\dots,l_k) ]=s'  $.
	Assume that $M(s')$ is lca-conform. 
	Then $\nu_M (s',u.1)$ and $\nu_M (s'  ,u.2)$ must be disjoint  otherwise  $M(s') / \nu_M (s' ,u.1)$
	contains some node labeled by $q'(u.2)$ or vice versa and hence  $M(s')$ would not be
	lca-conform.  Note that some nodes with label  $q_j'(u.j)$,  $j=1,2$, must occur in $M(s') $ 
	otherwise nodes with labels  of the form $q_1(u_1)$   and $q_2(u_2)$ cannot occur in $M(s)$.
	Clearly only descendants of $\nu_M (s',u.j)$ occurring $V(M(s')$ process the subtree $s/u.j$, $j=1,2$.
	Hence only descendants of 
	$\nu_M (s',u.j)$ occurring $V(M(s))$ can have labels of the form $q_j (u_j)$. This implies that $\nu_M (s,u_j)$ is a descendant of
	$\nu_M (s',u.j)$, which means that  $\nu_M (s,u_1)$ and $\nu_M (s,u_2)$  are disjoint contradicting (*).\qed
\end{proof}

We now show how to decide whether or not $M$ is lca-conform.
By Lemma~\ref{decision aux}, if $M$ is not lca-conform a 
context $c$, $l_1,\dots,l_k \in L$, and $a\in \Sigma_k$, $k\geq 2$, exists such that 
$t=M(c  a[l_1,\dots,l_k])$ is not lca-conform, 
i.e.,  leaves  $v_1, v_2, v_3 \in V(t)$ labeled $q_1(u)$, $q_2 (u)$ and $q_3 (u')$, respectively, with $u\neq u'$
such that the lowest common ancestor of $v_1$ and $v_2$ is an ancestor of $v_3$ exist.
For these leaves with these labels and this  ancestor relation to occur in $t$ one of the following must hold
\begin{itemize}
	\item  $q\in Q$  exists 
	such that  leaves  $v'_1, v'_2, v'_3$ labeled $q_1(\hat{u})$, $q_2 (\hat{u})$ and $q_3 (\hat{u}')$, respectively, occur in $\makro{q}(a[l_1,\dots,l_k])$ where  the lowest common ancestor of $v'_1$ and $v'_2$ 
	is an ancestor of $v'_3$
	\item pairwise occurring states $q$, $q'$  exists such that  $q_1(\hat{u})$ and $q_3 (\hat{u}')$ occur in $\makro{q}(a[l_1,\dots,l_k])$
	and   $q_2 (\hat{u})$  occurs in $\makro{q'}(a[l_1,\dots,l_k])$
	\item there exists a triplet $((q, q') q'')$  such that
	$q_1(\hat{u})$, $q_2 (\hat{u})$ and $q_3 (\hat{u}')$ occur in $\makro{q}(a[l_1,\dots,l_k])$
	$\makro{q'}(a[l_1,\dots,l_k])$ and $\makro{q''}(a[l_1,\dots,l_k])$, respectively,
\end{itemize}
where $\hat{u},\hat{u}' \in V(a[l_1,\dots,l_k])$  with $\hat{u}\neq \hat{u}'$.
In the following we say that a tuple $((q_1,q_2),q_3)$ is a \emph{triplet} if there exists some context $c$ and a la-state $l$ such that three distinct
nodes $v_1$, $v_2$ and $v_3$ with labels $q_1(u)$, $q_2(u)$ and $q_3(u)$, respectively, occur in ${M}(c  l)$ and
the lowest common ancestor of $v_1$ and $v_2$ is an ancestor of $v_3$ as well.

Determining the set of all triplets occurring in $M$ is analogous to determining $\text{PO}$ and can be done in polynomial time. 
Note that determining the 	 lowest common ancestor of $v_1$ and $v_2$ and whether this node is an ancestor of $v_3$
essentially consists of determining the longest common prefix of $v_1$ and $v_2$ and checking whether this string is a prefix of $v_3$. 

To decide whether or not $M$ is lowest common ancestor, it suffices to check whether a state, pairwise occurring states or a triplet 
exists such that one of the conditions above is satisfied for some $a\in \Sigma_k$, $k\geq 2$, and some $l_1,\dots,l_k \in L$.
This can be done in polynomial time.

\subsection{Proof of Lemma~\ref{lca-tree bound}}
In the following we denote by $\eta$ the maximal height of any right hand side of $M$.
Before we prove Theorem~\ref{lca-tree bound} consider the following definition.

Consider some output tree $M(s)$ where $s\in T_\Sigma [L]$. Let $v$ be a node such that   $M(s)[v]=d \in \Delta$.
We call the node $u$ of $s$ the 
\emph{origin of} $v$ in ${M}(s)$ if for  $\delta^*(s/u)=l$
\begin{enumerate}
	\item $v \notin V({M}(s[u\leftarrow l]))$  or ${M}(s[u\leftarrow l])/v \in Q(L)$  and
	\item ${M}(s[u.i\leftarrow  l_i \mid i\in [k]] ) [v] =d $ where $s[u] \in \Sigma_k$, $k\geq 0$, and
	$\delta^*(s/u.i)=l_i$.
\end{enumerate}
If $v$ is a node occurring in $A(l')$ where $l'=\delta^*(s)$ then we say that the origin of $v$ in ${M}(s)$ is NULL.
Informally,  one can consider the origin of a node $v$ to be the node of $s$ which ``creates'' the $\Delta$-node $v$.

Let $c$ be a context and $l$ be a la-state. Let  $t={M} (c l) / \nu (u)$. 
Let $u_0$ be the origin of $\nu (u)$. We show that if  $M$ is lca-conform
then  only the nodes that are descendants of $u_0$ and ancestors of $u$ are important for the generation of $t$.
All other descendants of $u_0$  do not influence the generation of $t$, i.e.,
removing descendants of $u_0$ that are not  ancestors of $u$ from $cl$ does not affect $t$.  We prove this statement in the following.

\begin{lemma}\label{lca-tree origins}
	Let $M$ be lca-conform.
	Let $c$ be a context and $l$ be an la-state such that $\delta^* (cl)\in F$.
	Let $u$ be the node  such that $cl [u]=l$.
	Let  $t={M} (c l) / \nu (u)$ and $u_0$ be the origin of  $\nu (u)$.
	Let   $u'$ be an arbitrary descendant  of $u_0$ that is not an ancestor of $u$. 
	Let $s$ be the tree that results from 
	replacing the subtree rooted at $u'$ by $\delta^*(cl/u')$. Then 
	$t ={M} (s) / \nu (u)$.
\end{lemma}

\begin{proof}
	First, consider the following.
	As $u'$ is not an ancestor of $u$ clearly the node $u$ still occurs in $s$ and $s[u]=l$ holds. 
	Therefore, we deduce that a node $v$  is labeled by $q(u)$, where $q\in Q$, in  ${M}(s)$ if and only if $v$
	is labeled by $q(u)$ in ${M}(c l)$. Hence $\nu_M (c l, u)=\nu_M (s, u)$.
	
	Secondly, it obviously holds that $cl[u'] \in \Delta$. Therefore it is easy to see that if no node in $t$
	has origin $u'$ then $t ={M} (s) / \nu (u)$.
	Assume that $t \neq {M} (s) / \nu (u)$ and thus that some descendant  $v'$ of $\nu_M (c l, u)$  has origin $u'$.
	Clearly, we can assume that $v'$ has label of the form $q(u')$ in $M(s)$. If such a node $v'$
	does not exist then no node in $t$ can have origin $u'$.
	This contradicts the lca-conformity of $M$ as $\nu_M (c l, u)=\nu_M (s, u)$,
	 $v'$ is a descendant of $\nu_M (c l, u)$
	 and $M(s)[v']=q(u')$ with $u'\neq u$. \qed
\end{proof}

\noindent
In the following we denote by $V(c)_{u_0,u}$ the set of all nodes in $V(c)$ that are descendants of $u_0$ and children of some ancestor of $u$, but no
ancestors of $u$ themselves. 

\begin{corollary}\label{lca-tree origins corollary}
	Let $M$ be lca-conform.
	Let $c$ be a context and $l$ be a la-state such that $\delta^* (cl)\in F$.
	Let $u$ be the node  such that $cl [u]=l$.
	Then $$M (c l) /\nu (u) = M (cl [u'\leftarrow l' \mid u' \in V(c)_{u_0,u} \text{ and }
	\delta^*(cl /u')=l' ] )/ \nu (u)$$
	where $u_0$ is the origin of  $\nu (u)$.
\end{corollary}

\noindent
We now prove our main claim.

\begin{lemma}
	Let $M$ be lca-conform and zero output twinned.
	Let $c$ be a context and $l$ be an la-state such that $\delta^* (cl)\in F$. Let $u$ be the node such that $cl [u]=l$. Then ${M}(cl)/ \nu(u)$
	has height  at most $(|Q|^2 +1) \eta$.
\end{lemma}

\begin{proof}
	 Assume  that the height of ${M}(cl)/ \nu(u)$ 
	exceeds $(|Q|^2 +1) \eta$. Then clearly $\nu (u)$ cannot be a leaf.
	
	Assume that the origin of $\nu (u)$ is not NULL.
	Let the node $u_0$ be the origin of $\nu (u)$. 
	Consider $$s=cl [u'\leftarrow l' \mid u' \in V(c)_{u_0,u} \text{ and }
	\delta^*(cl /u')=l'].$$
	As $M$ is lca-conform , it is sufficient to show that ${M}(s)/ \nu(u)$
	has height  at most $(|Q|^2 +1) \eta$ 	due to Corollary~\ref{lca-tree origins corollary}.
	Note that $\nu_M(cl,u) = \nu_M(s,u)$.
	Let $u_1,\dots u_m$ be  descendants of $u_0$  such that
	$u_{i}$  is a child node of $u_{i-1}$ for $i\in [m]$ and
	$u_m=u$.
	In the following  we denote by  $s_{i}$ the tree 
	$s [u_i\leftarrow \delta^*( s/u_i)]$.

	As $u_0$  is the origin of $\nu (u)$ it is clear that ${M}(s_1)/\nu(u) $ has height at most $\eta$.
	Additionally it holds that ${M}(s_1)[\nu (u)] \in \Delta$. Therefore there exist two distinct descendants  
	$v_1$ and $v_2$ of $\nu$ labeled $q_1(u_1)$ and $q_2(u_1)$ in ${M}(s_1)$, respectively,
	such that  $\makro{q_1}( s /u_1)$ and  $\makro{q_2}( s /u_1))$
	are not ground, otherwise $\nu (u)$ would not be the lowest common ancestor of all nodes with label of the form $q(u)$
	in ${M}(s)$. 
	
	As the height of  ${M}(cl)/ \nu(u)$
	exceeds $(|Q|^2 +1) \eta$ and ${M}(s_1)/\nu (u) $ has height at most $\eta$ we deduce that a descendant $v$ of 
	$\nu (u)$  labeled $q(u_1)$  in ${M}(s_1)$ must exists such that $\makro{q}( s /u_1)$ has height
	greater than $|Q|^2  \eta$~$(*)$.  W.l.o.g we assume $v\neq v_1$. Note that $q$ and $q_1$ are pairwise occurring.
	
	The only difference between  the trees $s_{i}$ and $s_{i+1}$ is that $s_{i} /u_i=l \in L$
	and $s_{i+1} /u_i=a(l_1,\dots,l_k)$ with $a\in \Sigma_k$, $l_1,\dots,l_k$, such that $\delta(a,l_1,\dots,l_k)=l$.
	Therefore the height difference of  ${M}(s_{i})/v$ and $ {M}(s_{i+1} )/v$ 
	is at most  $\eta$. 
	Clearly, 	${M}( s_{i} )/v=\makro{q}_{u_1}(s_{i}/u_1)$ for all $i\in [m]$. Hence 
	the height difference between $\makro{q}(s_{i}/u_1)$ and $\makro{q}(s_{i+1}/u_1)$
	is bounded by $\eta$~$(*')$.
	
	Note that by definition $s=s_m$. Let $n=|Q|^2$.
	Due to $(*)$ and $(*')$ we deduce that $n=|Q|^2 < m$.  Additionally,
	it follows that 
	nodes $v_{\iota_1},\dots,v_{\iota_n}$ where $1< \iota_1< \dots < \iota_{n}< m$ exist such that
	 $v_{\iota_{\tau}}$ has label of the form $q(u_{\iota_\tau})$ in $\makro{q}(s_{\iota_{\tau}}/u_1)$ and
	$v_{\iota_{\tau}}$ is a proper ancestor of  $v_{\iota_{\tau+1}}$. 
	Thus, as $\makro{q_1}( s/u_1)$ is not ground, states $q'_{\iota_1},\dots,q'_{\iota_n}$ exist such that
	\[ (q,q_1)\vdash^{s_{\iota_1}/u_1} (q_{\iota_1}, q'_{\iota_1} ) \vdash^{s_{\iota_2}/u_{\iota_1}} \dots  \vdash^{s_{\iota_n}/u_{\iota_{n_1}}} (q_{\iota_{n}}, q'_{\iota_{n}} ). \]
	This in turn means that $\gamma$ and $\vartheta$ with $0\leq \gamma <\vartheta \leq n$  exist such that $q_{\iota_{\gamma}}=q_{\iota_{\vartheta}}$ and
	$q'_{\iota_{\gamma}}=q'_{\iota_{\vartheta}}$ which means 
	$(q_{\iota_{\gamma}},q'_{\iota_{\gamma}}) \vdash^{s_{\iota_\vartheta}/u_{\iota_\gamma}} (q_{\iota_{\vartheta}},q'_{\iota_{\vartheta}})=(q_{\iota_{\gamma}},q'_{\iota_{\gamma}})$.
	This contradicts the zero output twinned property of $M$.
	
	Consider the case that the origin of $\nu (u)$ is NULL. Let $\hat{l}=\delta^* (c l)$. Clearly, $M(\hat{l})/\nu (u)$ has height of at most
	$\eta$. Furthermore it is obvious that two distinct descendants of $\nu (u)$ have label of the form $q(\lambda)$ in $M(\hat{l})$.
	The argument is therefore analogous to the previous case.\qed
\end{proof}

\subsection{Construction for La-Transducer and Proof of Correctness}
We give a construction which yields  a linear la-transducer $N$ equivalent to a given canonical earliest la-transducer $M$ 
if $M$ is zero output twinned and lca-conform.

Assume in the following that there is an la-state $l\in F$ such that  $A(l)$ is  not ground (the case that $A(l)$ is ground
for all $l\in F$ is trivial). 
By Lemma~\ref{auxilliary appendix} the output of $M$ is ``ahead'' of the output of any equivalent la-transducer with the same la-automaton. 
 Therefore, the basic idea is the same as in~\cite{DBLP:conf/icalp/ManethS20} and \cite{DBLP:journals/tcs/EngelfrietMS16}: we store the ``aheadness'' of $M$ compared to $N$ in the states of $N$.
The  states of $N$ are given by $\langle t \rangle$, where $t$ is a non-ground tree
in $T_\Delta [Q]$. Furthermore we demand that   the root of $t$
is the lowest common ancestor of all leaves of $t$ labeled by  symbols in $Q$.
For  such a state $\langle t\rangle$ 
we write ${t}[\leftarrow b]$ instead of $t[q\leftarrow q(b) \mid q\in Q]$ where $b$ is some symbol
for  better readability. 

We define  $N$ inductively starting with the mapping $A'$.
Let $l\in F$. Then we define $A'(l)=c \langle t_0 \rangle (x_1)$, where $c$ is a context and $t_0\in T_\Delta [Q]$,  such that  $c t_0 [\leftarrow x_1]   =A(l)$. 
Note that $t_0 [\leftarrow x_1] $ is the subtree  of $A(l)$ rooted at the lowest common ancestor of all nodes
with label of the form $q(x_1)$.

We now define the rules of $N$.
Let  $\langle t \rangle$ be a state of $N$ where $t=p[x_i\leftarrow q_i \mid q_i\in Q, i\in [n]]$ and $p\in T_\Delta [X_n]$. Let $a\in \Sigma_k$ with $k\geq 0$  and $l_1,\dots ,l_k$ be la-states such that
for all $i\in [n]$, $\text{rhs}_M(q_i,a,l_1,\dots,l_k)$ is defined.
Then, we define
$\langle t \rangle (a(\stat{x_1}{l_1},\dots,\stat{x_k}{l_k})) \to p'[x_j\to \langle t_j \rangle (x_j)\mid j\in [k]  ]$
where $p'\in T_\Delta [X_k]$   such that each variable $x_j$, $j\in [k]$, occurs at most once in $p'$ and
\[ p'[x_j\leftarrow t_j[\leftarrow  x_j]\mid j\in [k]  ]= p[ x_i \leftarrow \rhs{q_i,a,l_1,\dots,l_k}  \mid i\in [n]]\]
and $t_j[\leftarrow  x_j]$ is the subtree of
$ p[ x_i \leftarrow \rhs{q_i,a,l_1,\dots,l_k}  \mid i\in [n]]$ rooted at the lowest common ancestor of all nodes
with labels of the form $q(x_j)$.
In the following we show that  states of the form  $\langle t\rangle$ indeed  store the ``lead'' of $M$ compared to $N$.

\begin{lemma}\label{aux partial lemma}
	Assume that in our construction some state $\langle t \rangle$  has been defined at some point.
	Then there exists a context $c$, a  la-state  $l$ and a node $u$ such that $cl[u]=l$ and
	${M}(cl)/\nu (u)=t[\leftarrow u]$.
\end{lemma}

\begin{proof}
	We prove our claim by structural induction.
	If  $\langle t \rangle$  has been defined then there exists a la-state $l\in F$ and a context $c$ such that
	such that $A'(l)=c \langle t \rangle (x_1)$
	or  there exists some state $\langle t' \rangle$, la-states $l_1,\dots,l_k$,  and some $a\in \Sigma_k$, $k>0$, such that  $\langle t \rangle (x_j)$ occurs in
	$\text{rhs}_N (\langle t' \rangle, a ,l_1,\dots, l_k)$  and $j\in [k]$. As the former case is trivial consider the latter case.
	
	Let $t'=p[q_1 ,\dots,q_{n}]$ and  $\delta (a,l_1,\dots,l_k)=l_0$. As $\text{rhs}_N (\langle t' \rangle, a ,l_1,\dots, l_k)$
	is defined, by construction  $\text{rhs}_M (q_i,a,l_1,\dots,l_k)$ is defined for all $i\in [n]$. Thus,
	the la-conformity of $M$ yields
	 $\rho (q_i)=l_0$ for all $i\in [n]$. 
	Let $\text{rhs}_N (\langle t' \rangle, a ,l_1,\dots, l_k)= p'[ x_j\leftarrow \langle t_j \rangle (x_j) \mid j\in [k] ]$. 
	By construction
	\begin{equation}\label{aux equation}
	p'[ x_j\leftarrow t_j[\leftarrow x_j] \mid j\in [k] ] =
	p[ x_i \leftarrow \rhs{q_1,a,l_1,\dots,l_k} \mid i \in [n]].
	\end{equation}
	By induction, a context $c'$, a la-state $l'$ and a node $u'$ exist such that ${M}(c'l')/ \nu(u')
	=t'[\leftarrow u']= p[q_1 (u'),\dots,q_{n} (u')]$ and $c'l'[u']=l'$.
	Recall that  $\rho (q_i)=l_0$ for all $i\in [n]$ and hence by Corollary~\ref{appendix corollary}, $l'=l_0$.
	Therefore, for  $s=c'l' [u' \leftarrow a[l_1, \dots, l_k]]$ the following equation   \begin{multline}\label{aux equation 2}
	{M}(s) / \nu(u')= {M}(cl') [q(u') \leftarrow
	\makro{q}_{u'} (a[l_1, \dots, l_k])  \mid q\in Q] / \nu(u') =\\
	p[ x_i \leftarrow \makro{q_i}_{u'} (a[l_1, \dots, l_k]) \mid i\in [n] ]
	\end{multline}
	holds.
	Due to Equation~\ref{aux equation} it follows easily that, 
	\begin{equation*}
	p[ x_i \leftarrow \makro{q_i}_{u'} (a[l_1, \dots, l_k]) \mid i\in [n] ]=
	p'[x_j\leftarrow t_j[\leftarrow u'.j] \mid j\in [k] ]. 
	\end{equation*}
	Note that ${M}(s) / \nu(u')$ contains all nodes of $M(s)$ with label in $Q(V)$.
	Therefore it follows that
	$t_j[\leftarrow u.j]={M}(s) / \nu( u'.j) $ holds. 
	Thus, it is easy to see that our claim holds. \qed
\end{proof}

Due to Lemmas~\ref{lca-tree bound} and~\ref{aux partial lemma},
the height of any tree $t$ such that $\langle t \rangle$ is a state of $N$ is bounded.
Accordingly, our construction  terminates.
We now show that our construction is well defined.

\begin{lemma}\label{not unsuccessful}
 Let $\langle t \rangle$ with $t=p[q_1,\dots,q_n]$ and let $\text{rhs}_M (q_i,a,l_1,\dots,l_k)$
 be defined  for all $i\in [n]$.
	Then,  $p'\in T_\Sigma [X_k]$ exists such that
	$$  p'[x_j\leftarrow t_j[\leftarrow  x_j]\mid j\in [k]  ]= p[ x_i \leftarrow \rhs{q_i,a,l_1,\dots,l_k}  \mid i\in [n]],$$
	where $t_j[\leftarrow  x_j]$ is the subtree of
	$ p[ x_i \leftarrow \rhs{q_i,a,l_1,\dots,l_k}  \mid i\in [n]]$ rooted at the lowest common ancestor of all nodes
	with labels of the form $q(x_j)$, 
	and each variable $x_j$, $j\in [k]$, occurs at most once in $p'$.
\end{lemma}

\begin{proof}
	By Lemma~\ref{aux partial lemma} a context $c$ and an la-state $l$ such that
	${M}(cl)/\nu (u) =t [\leftarrow u] =p[q_1 (u),\dots,q_n (u)]$.
	Let $s=c [u\leftarrow a (l_1,\dots ,l_k) ]$. 
	Then  ${M}(s)={M}(s') [q(u) \leftarrow \makro{q}_u( a [l_1,\dots ,l_k]) \mid q\in Q ]$. 
	Hence $${M}(s) /\nu_M (cl,u) = p[x_i \leftarrow  \makro{q_i}_u( a (l_1,\dots ,l_k)) \mid i [n]].$$
	Note that ${M}(s) /\nu_M (cl,u)$ contains all nodes of $M(s)$ labeled by a symbol in $Q(V)$.
	Therefore,   
	for all $i\in [k]$, $\nu_M (cl,u)$ is an ancestor of $\nu_M (s,u.i)$.
	Additionally, the lca-conformity of $M$ yields that  $\nu_M (s,u.i)$ and $\nu_M (s,u.i')$ are disjoint if
	$i\neq i'$.  
	Therefore it follows easily that  $p'\in T_\Sigma [X_k]$  exist such that
	$  p'[x_j\leftarrow t_j[\leftarrow  x_j]\mid j\in [k]  ]= p[ x_i \leftarrow \rhs{q_i,a,l_1,\dots,l_k}  \mid i\in [n]]$,
	where $t_i \in T(Q)$ such that ${M}(s) /\nu (u.i) =t_i[\leftarrow u.i]$,
	and each variable $x_i$, $i\in [n]$, occurs at most once in $p'$. \qed
\end{proof}

\noindent
We now prove that the la-transducer $N$ our construction yields is equivalent to $M$.

\begin{lemma}\label{equal}
	If the construction terminates successfully then  it  yields an la-transducer $N$ which is linear and equivalent to $M$.
\end{lemma}
\begin{proof}
	Clearly our construction yields a linear la-transducer. Hence, we only need to show that $M$ and
	$N$ are equivalent.  Note that by construction both la-transducers have the same la-automaton.
	We prove the equivalence by structural induction.
	
	In particular we show that for an arbitrary tree  $s=\bar{s}[x_1\leftarrow l]$ where $\bar{s}\in T_\Sigma \{X_1\}$ and  $l\in L$, such that $\delta^* (s)\in F$, ${M}(s)={N}(s) [\langle t \rangle (u) \leftarrow t[\leftarrow u]\mid \langle t \rangle \in Q', u\in V(s)]$. Obviously this statement
	holds if $s=l$ and $l\in F$. 
	
	 Let $u$ be a node of $s$ such that $s/u=a[s_1,\dots,s_k]$ where $a\in \Sigma_k$, $k\geq 0$, 
	 $s_i= l$ for some $i\in [k]$ and $s_j \in T_\Sigma$ if $j\neq i$.
	 Let  $\delta^*(a[s_1,\dots,s_k])=l'$ and $s'=s[u'\leftarrow l']$. 
	By induction $${M}(s')={N}(s') [\langle t' \rangle (u') \leftarrow t [\leftarrow u'] \mid \langle t' \rangle \in Q'].$$
	Let $\nu$ be the node labeled $\langle t' \rangle (u') $ in $N(s)$. As $N$ is linear there is only one such node.
	 Clearly $M(s)/\nu = t' [\leftarrow u']$.
	Let $t'=p[q_1,\dots q_n]$. By construction,  $\langle t' \rangle (a(\stat{x_1}{l_1},\dots,\stat{x_k}{l_k})) \to p'[x_j\to \langle t_j \rangle (x_j)\mid j\in [k]  ]$
	where $p'\in T_\Delta [X_k]$   such that 
	\[ p'[x_j\leftarrow t_j[\leftarrow  x_j]\mid j\in [k]  ]= p[ x_i \leftarrow \rhs{q_i,a,l_1,\dots,l_k}  \mid i\in [n]]\ (*).\]
	
	Clearly, $M(s)/\nu = p[x_i\leftarrow \makro{q_i}_{u'} (a[s_1,\dots,s_k]) \mid i \in [n] ]$ and
	${N}(s)/\nu = \makro{ \langle t' \rangle}_{u'} (a[s_1,\dots,s_k]) =p'[x_j\leftarrow \makro{ \langle t_j \rangle}_{u'.j}  (s_j) \mid j\in [k]  ] $.
	Note that $M(s)/\nu$ and $N(s)/\nu$ contain all nodes of $M(s)$ and $N(s)$ not labeled by symbols in $\Delta$.
	Recall that $s_i=l\in L$. Thus, $ \makro{ \langle t_i \rangle}_{u'.i}  (s_i) = \langle t_i \rangle (u'.i)$
	and nodes with label of the form $q(u'.i)$ occur in $M(s)$.
	Clearly a node is labeled by $ \langle t_i \rangle (u'.i)$ in $N(s)/\nu$ iff that node is labeled by $\langle t_i \rangle (x_i)$ in $p'[x_j\leftarrow t_j[\leftarrow  x_j]\mid j\in [k]  ]$.
	Analogously, a node ha label of the form $q(u'.i)$ in $M(s)$ iff that node has label of the form $q(x_i)$ in
	  $p[ x_i \leftarrow \rhs{q_i,a,l_1,\dots,l_k}  \mid i\in [n]]$.
	Therefore,  ${M}(s)/\nu \neq {N}(s)\nu [\langle t' \rangle (u'.i) \leftarrow t[\leftarrow u'.i] ]$ contradicts $(*)$.
	Hence our claim follows. \qed
\end{proof}

The construction and Lemmas~\ref{zero output twinned} and~\ref{lca-conform} yield the following theorems.

\begin{theorem}
	Let $M$ be a  canonical earliest la-transducer. An equivalent linear la-transducer $N$ with the same la-automaton as 
	$M$  exists if and only if 
	$M$ is zero output twinned  and lca-conform.
\end{theorem}

\begin{theorem}\label{partial main result}
	Let $M$ be an la transducer. 
	It is decidable  whether or not an 
	equivalent linear la-transducer $N$ with the same la-automaton as $M$  exists, 
	and if so $N$ can be effectively constructed. 
\end{theorem}

\subsection{Proof of Theorem~\ref{linear main result}}
\begin{proof}
First, note that  if a linear transducer (without look-ahead) $T$ equivalent to $M$ exists then there is a linear la-transducer
$N$ equivalent to $M$ which has the same la-automaton as $M$. 
By Theorem~\ref{partial main result} it is decidable for a given la-transducer whether or not
an equivalent linear la-transducer  with the same la-automaton  exists and 
if so then this linear la-transducer  can  be constructed.
Hence, if no such transducer exists we return ``no''. Otherwise, we
use Theorem~20 of \cite{DBLP:conf/icalp/ManethS20} to decide whether or not a linear transducer equivalent to $N$ exists and
in the affirmative case we construct such a transducer. \qed
\end{proof}

\subsection{Proof of Lemma~\ref{hom lemma 1}}

\begin{proof}
	Assume the contrary. 
	Then $a\in \Sigma_k$, $k>0$, must exist such that $t_1={T} (a(x_1,\dots,x_k))$ is not subtree conform to $A$.
	Let $t_2=h(a(x_1,\dots,x_k))$.
	
	By 	Corollary~\ref{auxilliary appendix 2}, $t_1= t_2[ v\leftarrow t_1/v \mid v\in V_{\neg \Delta}(t_2)]$. Therefore if for all 
	$v\in V_{\neg \Delta}(t_2)$, either 
	(1)~$t_1/v\simeq A$  or (2)~$t_1/v \in T_\Delta$ holds then $t_1$ would be subtree conform to $A$.
	Therefore, there must be a node $v'\in V_{\neg \Delta}(t_2)$
	such that  (1)~$t_1/v'\not\simeq A$  and (2)~$t_1/v'$ is not ground. Let 	$t_2/v'=h(i)$, $i\in [k]$.
	For arbitrary trees $s_1,\dots,s_k$, clearly
	\begin{multline}\label{central equation}
	t_1[q(x_j) \leftarrow \makro{q}(s_j)\mid q\in Q, j\in [k]] /v'= {T} (a(s_1,\dots,s_k))/v'=\\
	h(a(s_1,\dots,s_k))/v'=h(s_i)=  T(s_i) = A[q'(x_1) \leftarrow \makro{q'}(s_i)\mid q'\in Q]
	\end{multline}
	First we show that $V(t_1/v') \subseteq V(A)$. Assume that a node in $V(t_1/v') \setminus V(A)$ exists, then, a proper ancestor $v_A$ of that node exists such that $A [v_A] =q_A(x_1)$ otherwise Equation~\ref{central equation} would not hold.
	Clearly  $t_1/v_1 [v_A]=d\in\Delta$ holds in this case.
	If a node $v_A$ such that  $A [v_A] =q_A(x_1)$ and $t_1/v_1 [v_A]=d$ exists,then  Equation~\ref{central equation}  implies that $q_A$ generates trees with root label $d$ for arbitrary inputs contradicting the earliest property of $T$.
	Therefore we deduce that $V(t_1/v') \subseteq V(A)$ and that if $v\in V(t_1/v')$ such that $t/v_1 [v] =d\in \Delta$ then
	$A [v] =d$.
	Analogously, it follows that $V(t_1/v') \supseteq V(A)$  and that if $A [v] =d\in \Delta$ then $t/v_1 [v] =d$, 
	which means $V(t_1/v') = V(A)$ and 	$A [v_A] =d \in \Delta$ if and only if $t_1/v_1 [v_A]=d\in\Delta$.
	
	Let $v_{A} \in V(A)$ such that $t_1/v'[v_A]=q (j)$, $j\in [k]$, and $A[v_A]= q'(x_1)$.
	As $s_1,\dots,s_k$ are arbitrary consider the case in which these trees are identical. 
	Then Equation~\ref{central equation} yields $\makro{q}(s_1)=\makro{q'} (s_1)$  for arbitrary  $s_1\in T_\Sigma$. Hence $q= q'$
	as $T$ is canonical earliest.
	
	By  Equation~\ref{central equation}	$$t_1/v'[q(x_j) \leftarrow \makro{q}(s_j)\mid q\in Q, j\in [k]]=   A[q'(x_1) \leftarrow \makro{q'}(s_i)\mid q'\in Q].$$
	We now show that if $q(x_j)$ occurs in $t_1/v'$ then $i=j$. Assume the contrary.
	Consider the case that  $s_1,\dots,s_k$ are mutually distinct. 
	Due to previous considerations clearly if $t_1/v'[v_A]=q (j)$ then $A[v_A]= q(x_1)$.
	Thus
	Equation~\ref{central equation} yields $\makro{q}(s_j)=\makro{q'} (s_i)$ meaning that the output of $q$ and $q'$ do not depend on the input.
	This contradicts the earliest property of $T$. 
	Thus, we have shown that no node with label of the form $q(j)$ can occur in $t_1/v'$ such that $j\neq i$.
	Therefore, there is  clearly a bijection such that $t_1/v'\simeq A$. \qed
\end{proof}

\subsection{Full Proof of Lemma~\ref{hom lemma 2}}

\begin{proof}
	We construct the homomorphism $h$ as follows. Let $a\in \Sigma_k$, $k\geq 0$.
	If ${T} (a(x_1,\dots,x_k))$ is ground, then we define $h (a(x_1,\dots,x_k))= 	{T} (a(x_1,\dots,x_k))$.
	Otherwise,
	our premise yields ${T} (a(x_1,\dots,x_k))=t[x_j\leftarrow A({j}) \mid j\in [k] ]$ 
	for a suitable $t\in T_\Delta [X_k]$ 
	where $A(j)=A[
	q(x_1)\leftarrow q(j) \mid q\in Q]$.
	Then we define  $h (a(x_1,\dots,x_k))= 
	t[ x_j\leftarrow h({j}) \mid j\in [k] ]$.
	
	The equality of $h$ and $T$  is shown by induction on the structure of the input trees.
	Let $s=a(s_1,\dots,s_k)$  with $a\in \Sigma_k$ and $s_1,\dots,s_k \in T_\Sigma$.
	By our construction $h(a(s_1,\dots,s_k))=t[ x_j\leftarrow h(s_{j}) \mid j\in [k]  ]$,
	where  $t$ is the pattern for which ${T} (a(x_1,\dots,x_k))=t[
	x_j\leftarrow A({j}) \mid j\in [k] ]$ holds.  
	By  induction  $h(s_{j})=	{T} (s_{j})$ holds for all $j\in [k]$  and hence 
	$h(a(s_1,\dots,s_k))=t[x_j\leftarrow T(s_{j}) \mid j\in [k] ]$.
	The latter tree equals ${T}(a(s_1,\dots, s_k))$ by the choice of $t$ and the
	definition of ${T}$.
	Therefore our claim follows.\qed
\end{proof}

\subsection{Deciding Subtree conformity to $A$ and Time Complexity of the Construction in  Lemma~\ref{hom lemma 2}}

\begin{theorem}
	It is decidable in polynomial time whether or not 
	${T}(a(x_1,\dots,x_k))$ is subtree conform to $A$ for all $a\in \Sigma$
	where $T$ is a total transducer $T$ and $A$ is its axiom.
	 In the affirmative case,
	an equivalent homomrphism $h$ can be constructed in  polynomial time as well.
\end{theorem}

\begin{proof}
	Whether or not ${T} (a(x_1,\dots,x_k))$ is subtree conform to $A$ for all $a\in \Sigma$ can be determined as follows.
	For all $a\in \Sigma$ we analyze the structure of 	$t={T} (a(x_1,\dots,x_k)$.
	Consider all nodes $v$ of $t$ such that  $\text{height}(t/v)= \text{height}(A)$.
	If  $t$ is subtree conform to $A$ then clearly 
	$t/v$ must be either ground or $t/v \simeq A$ must hold. Hence if $t/v$ is   non-ground or $t/v \not\simeq A$ then clearly 
	${T} (a(x_1,\dots,x_k))$ cannot be subtree conform to $A$.  (Note that if  $t$ is subtree conform to $A$  and $\text{height}(t)< \text{height}(A)$
	then $t$ must be ground.)
	Determining all nodes $v$ of $t$ such that  $\text{height}(t/v)= \text{height}(A)$
	can clearly be done in  time linear in $|t|$.
	Checking whether $t/v$ is ground or whether    $t/v \simeq A$ can  be
	done in time linear in $|t/v|$. 
	Hence determining whether $t/v$ is ground or whether    $t/v \simeq A$ for all $v\in V(t)$ 
	such that $\text{height}(t/v)= \text{height}(A)$ can  be
	done in time linear in $|t|$.
	Note that all nodes $v$ such that  $\text{height}(t/v)= \text{height}(A)$
	are disjoint.  Hence the combined size of all such  subtrees $t/v$ 
	can obviously not exceed $|t|$.
	
	By definition $|\rhs{q,a} |\leq |T|$, where  $a\in \Sigma$, $q\in Q$, and $|A|\leq |T|$ hold.
	Additionally, it holds  that $A$ has no more than $|T|$ nodes with label of the form $q(x_1)$.
	Therefore $|t|=|{T} (a(x_1,\dots,x_k))|\leq |T|+|T|^2$ holds.
	
	Therefore, determining whether ${T} (a(x_1,\dots,x_k))$ is subtree conform to $A$ for all $a\in \Sigma$ can be done in time $O(|\Sigma| |T|^2)$.
	
	Note that constructing a homomorphism $h$ that is equivalent to a total transducer $T$ such that
	${T} (a(x_1,\dots,x_k))$ is subtree conform to $A$ for all $a\in \Sigma$ can be done in  time $O(|\Sigma| |T|^2)$ as well.
	For the construction in  Lemma~\ref{hom lemma 2} it is clearly sufficient to determine all 
	nodes $v$ of ${T}(a)$ such that ${T}(a)/v \simeq A$
	for all $a\in \Sigma$. \qed
\end{proof}

\subsection{Determining the total la-transducer $M'$}
In the following we show  that the existence of a non-total la-transducer $M$ that is equivalent to some homomorphism
implies that a total la-transducer $M'$ exists that is effectively equivalent to $M$. In particular we show that $M'$
can be obtained from $M$ by restricting its domain.
\begin{lemma}\label{partition}
	Let  $M$ be a non-total  transducer with regular-look ahead that is equivalent to a homomorphism $h$. Let $\Sigma$ be the input alphabet of $M$ and $h$. Then there exists a subset $\Sigma'$ of $\Sigma$ such that
	${M}$ is undefined for trees in $T_\Sigma \setminus T_{\Sigma'}$ while
	${M}$ restricted to $T_{\Sigma'}$ is total.
\end{lemma}

\begin{proof}
	As $M$ is non-total and equivalent to $h$, $h$ is clearly non-total as well. By definition $h$ is a single-state transducer. Therefore, $h$ being non-total means
	there is a subset $\Sigma''$ of $\Sigma$ such that $h(a_1(x_1,\dots,x_n))$ is undefined for all  $a_1\in \Sigma''$.
	On the other hand,  $h(a_2(x_1,\dots,x_n))$ is defined for all $a_2\in \Sigma'$, where  $\Sigma'=\Sigma\setminus \Sigma''$. This means that $h$ restricted to $T_{\Sigma'}$ is total,
	while $h$ is undefined for all trees in $T_\Sigma\setminus T_{\Sigma'}$. As $M$ and $h$ are equal,  ${M}$  is total if restricted to $T_{\Sigma'}$ and undefined for all trees in $T_\Sigma\setminus T_{\Sigma'}$ as well. \qed
\end{proof}

It is well-known that the domain of a transducer $M$ with regular look-ahead is effectively recognizable, i.e., there is a tree automaton
accepting $\text{dom}({M})$. \cite[Corollary 2.7]{DBLP:journals/mst/Engelfriet77}  Hence, it is   decidable whether there exists a
ranked alphabet $\Sigma'$ such that $\text{dom}({M})=T_{\Sigma'}$. If such $\Sigma'$
does not exist, then we return "no".

Otherwise, the la-transducer $M'$ is $M$ restricted to $T_{\Sigma'}$.
Obviously $M'$ is total.

\subsection{Considerations leading to Theorem~\ref{difference tree bound}}

But before we formalize our claim that on input $c$, $T$ cannot generate more nodes than ${M'} (c l_1) \sqcup {M'} (c  l_2)$ has.

\begin{lemma}\label{look ahead lemma}
	Consider a total transducer with regular look-ahead $M'$  that is equivalent
	to a total   transducer $T$.
	Let $l_1$ and $l_2$ be la-states of $M$. Let $c$ be a context over
	$\Sigma$. 
	Then  $V(T(c)) \subseteq V({M'} (cl_1) \sqcup {M'} (c l_2))$.
\end{lemma}

\begin{proof}
	In the following let $t={M'} (c l_1) \sqcup {M'} (c l_2)$.
	Let $s_{i}$ be a tree with la-state $l_i$ for $i=1,2$.
	Clearly ${M'} (c s_{i})={T} (c s_{i})$.
	
	Assume that  $V(T(c)) \not \subseteq V(t)$.
	Then a node in $V(T(c)) \setminus V(t)$ exists. 
	As ${M'} (c s_{i})={T} (c s_{i})$ for $i=1,2$, a proper ancestor $v$ of that node exists such that
 	$t[v] \in X$.
	Clearly $T(c)[v] \in \Delta$.
	By definition, the root symbols of ${M'} (c s_{1})/v$ and ${M'} (c s_{2})/v$ are different. Therefore, it is clear that either 
	${M'} (c s_{1}) [v] \neq T(c)[v]$ or ${M'} (c s_{2}) [v] \neq T(c)[v]$ holds.
	This in turn means that either ${M'} (c s_{1})\neq {T} (c s_{1})$
	or  ${M'} (c s_{2})\neq {T} (c s_{2})$ holds contradicting the equality of $T$ and $M$. This yields our claim. \qed
\end{proof}

\noindent
We now show how to determine the difference bound of $M'$.
\begin{lemma}
	Let $l_1,\dots,l_n$ be states of the la-automaton and $s_1,\dots ,s_n$ be trees of minimal height that reach these states.
Let $i\in[n]$ such that $\text{height}(M(s_i))\geq \text{height}(M(s_j))$, $j\neq i$, then, $\text{height}({M'} (s_i))$ is a difference bound of $M'$  if $M'$ is equivalent to a homomorphism 
\end{lemma}

\begin{proof}
	As $M'$ and $h$ are equivalent, there are clearly only  finitely
	many difference trees \cite[Corollary 28]{DBLP:journals/tcs/EngelfrietMS16} and hence a difference bound must exist.
	Note that by definition a homomorphism is a deterministic single state transducer.
	
	W.l.o.g. consider an arbitrary context $c$ over $\Sigma$ and the la-states $l_1,l_2$. 
	Let $v$ be a node of $t={M'} (c l_1) \sqcup {M'} (c l_2)$ with label in $X$ and hence $t_1 = {M'} (c l_1)/v$ and $t_2= {M'} (c l_2)/v$ be difference trees of $M'$.
	In the following, we  show that 
	$\text{height}(t_1)\leq \text{height}({M'} (s_1))$  and  $\text{height}(t_2)\leq \text{height}({M'} (s_{2})$
	hold.
	
	In the following let $\xi_1={M'} (c s_1)/v$ and $\xi_2={M'} (c s_{2})/v$.
	Clearly, $\text{height}(t_1)\leq \text{height}(\xi_1 )$ and $\text{height}(t_2)\leq \text{height}(\xi_2)$
	hold. 
	
	Let $h(c)= p[ h(x_1),\dots,h(x_1)]$ where $p\in T_\Delta \{X_n\}$. 
	By Lemma~\ref{look ahead lemma}, $V(p)=V(h(c)) \subseteq V(t)$. 
	Thus, the equality of $M$ and $h$ yields that some ancestor $v'$ of $v$ must exists such that $h(c)[v']= h(x_1)$.
	As $M'$ and $h$ are equivalent, ${M'} (c s_j)=h(c s_j)=	p[ h(s_{j}),\dots,h(s_{j})]$ holds for $j=1,2$.
	Therefore it follows easily that  
	$$\text{height}(\xi_j)\leq \text{height}(M' (c s_j) / v')=
	\text{height}(h (c s_j) / v')=\text{height}(h (s_j))
	$$  holds for $j=1,2$.
	As $M'$ and $h$ are equivalent, ${M'} (s_1)=h(s_1)$ and ${M'} (s_{2})= h(s_{2})$  hold. This yields  $\text{height}(t_1)\leq \text{height}({M'} (s_1))$ and $\text{height}(t_2)\leq \text{height}( {M'} (s_{2}) ) $.
	As $\text{height}(M(s_i))\geq \text{height}(M(s_j))$ if $j\neq i$,
	 clearly $\text{height}({M} (s_i))$ is a difference bound of $M$.  \qed
\end{proof}